\documentclass{article}

\usepackage[utf8]{inputenc}

\usepackage{xparse,xifthen}
\usepackage{mathtools} 
\usepackage{amsthm,amsxtra}
\usepackage{tikz,tikz-cd} 
\usepackage[auth-sc,affil-sl]{authblk}
\usepackage{thmtools}
\usepackage{enumitem}
\usepackage{braket}
\usepackage{slashed}

\usepackage{mathpazo}
\usepackage{MnSymbol}
\usepackage{stmaryrd}
\usepackage{prodint}

\usepackage{url}
\usepackage{hyperref}

\makeatletter
\newcommand{\pushright}[1]{\ifmeasuring@#1\else\omit\hfill$\displaystyle#1$\fi\ignorespaces}
\newcommand{\pushleft}[1]{\ifmeasuring@#1\else\omit$\displaystyle#1$\hfill\fi\ignorespaces}
\makeatother

\newcommand{\lo}[1]{\raisebox{-0.1ex}{$#1$}\,}
\newcommand{\loo}[1]{\raisebox{-0.2ex}{$#1$}\,}
\newcommand{\Lo}[1]{\raisebox{-0.3ex}{$#1$}\,}
\newcommand{\Loo}[1]{\raisebox{-0.4ex}{$#1$}\,}
\newcommand{\LO}[1]{\raisebox{-0.5ex}{$#1$}\,}
\newcommand{\LOO}[1]{\raisebox{-0.6ex}{$#1$}\,}

\newcommand{\R}{\mathbb R}
\newcommand{\C}{\mathbb C}
\newcommand{\N}{\mathbb N}

\newcommand{\T}{\mathbb T}

\newcommand{\eps}{\varepsilon}
\newcommand{\abs}[1]{\lvert #1 \rvert}
\newcommand{\Abs}[1]{\left\lvert #1 \right\rvert}

\newcommand{\D}{\mathrm{d}}
\newcommand{\I}{\mathrm{i}}
\newcommand{\e}{\mathrm{e}}
\newcommand{\Langle}{\left\langle}
\newcommand{\Rangle}{\right\rangle}
\newcommand{\Cb}[1]{C_\mathrm{b}\!\left(#1\right)}

\DeclareDocumentCommand{\EE}{ O{}mO{} }{
	\ifthenelse{\isempty{#3}}{	
		\mathbb E_{#1}\!\left[#2\right]}{
		\mathbb E_{#1}\!\left[#2\,\middle|\,#3\right]}
}

\newcommand{\mc}[1]{\mathcal{#1}}

\newcommand{\mr}[1]{\mathrm{#1}}

\declaretheorem[
	name=Definition, style=definition%
]{defn}
\declaretheorem[
	name=Hypothesis, style=definition
]{hypo}

\declaretheorem[
	name=Proposition, style=plain, sibling=defn
]{prop}
\declaretheorem[
	name=Lemma, style=plain, sibling=defn
]{lem}
\declaretheorem[
	name=Theorem, style=plain, sibling=defn
]{theo}
\declaretheorem[
	name=Corollary, style=plain, sibling=defn
]{coro}
\declaretheorem[
	name=Example, style=remark, sibling=defn
]{ex}
\declaretheorem[
	name=Remark, style=remark, sibling=defn
]{rk}

\newcommand{\propagatorlength}{10mm}
\newcommand{\propagatordistance}{6mm}
\newcommand{\defaultradius}{3mm}
\newcommand{\labelsize}{\scriptsize}

\NewDocumentCommand{\diagram}{m}{
\tikz[
	level 1/.style = 
		{level distance = \propagatorlength,
		sibling distance = \propagatordistance,
},
	every node/.style = 
		{
		minimum size = 6pt, 
		inner sep = 1pt,
		font = \labelsize},
	baseline = -.5ex 
]{#1}
}

\NewDocumentCommand{\coords}{mm}{(#1,#2)}

\NewDocumentCommand{\lloop} 	
	{O{\defaultradius} 
	r() 
	>{\SplitArgument{1}{,}}D(){0,0} 
	m} 
{
	\draw \coords#3 arc (0:360:#1);
	\node[draw, fill = white] (#2) at \coords#3 {$#4$};
}

\NewDocumentCommand{\rloop} 	
	{O{\defaultradius} 
	r() 
	>{\SplitArgument{1}{,}}D(){0,0} 
	O{}} 
{
	\draw \coords#3 arc (-180:180:#1);
	\node[draw, fill = white] (#2) at \coords#3 {$#4$};
}

\NewDocumentCommand{\dloop} 	
	{O{\defaultradius} 
	r() 
	>{\SplitArgument{1}{,}}D(){0,0} 
	m} 
{
	\draw \coords#3 arc (-270:90:#1);
	\node[draw, fill = white] (#2) at \coords#3 {$#4$};
}

\NewDocumentCommand{\uloop} 	
	{O{\defaultradius} 
	r() 
	>{\SplitArgument{1}{,}}D(){0,0} 
	m} 
{
	\draw \coords#3 arc (-90:270:#1);
	\node[draw, fill = white] (#2) at \coords#3 {$#4$};
}

\NewDocumentCommand{\loopdiagram}	
	{D(){\defaultradius}
	m} 
{ 
	\diagram{\uloop[#1](){#2}}
}

\NewDocumentCommand
	{\branchdiagram}	
	{s m O{}}
{\diagram{
	\IfBooleanTF{#1}{
		\branch*()(0,0){#2}[#3]
	}{
		\branch()(0,0){#2}[#3]
	}
}}

\NewDocumentCommand{\branch}
	{s 
	r() 
	>{\SplitArgument{1}{,}}D(){0,0} 
	m 
	>{\SplitArgument{2}{,}}O{}} 
{
	\IfBooleanTF{#1}{
	 	\linesaux{#2}#3{#4}#5
	}{
	 	\ldotslinesaux{#2}#3{#4}#5
	}
}
\NewDocumentCommand{\linesaux} 	
	{m 
	mm 
	m 
	mmm}	
{
	\begin{scope}[
		grow = right 
	]
 	\IfNoValueTF{#6}{
	 	\node[draw] (#1) at (#2,#3) {$#4$}
 		child {node[draw] {} edge from parent node[auto] {$#5$}};
 	}{\IfNoValueTF{#7}{
		 	\node[draw] (#1) at (#2,#3) {$#4$}
 			child {node[draw] {} edge from parent node[auto] {$#5$}}
 			child {node[draw] {} edge from parent node[auto] {$#6$}};
 		}{
		 	\node[draw] (#1) at (#2,#3) {$#4$}
 			child {node[draw] {} edge from parent node[auto] {$#5$}}
 			child {node[draw] {} edge from parent node[auto] {$#6$}}
 			child {node[draw] {} edge from parent node[auto] {$#7$}};
 		}
 	}
	\end{scope}
}
\NewDocumentCommand{\ldotslinesaux}
	{m 
	mm 
	m 
	mmm}	
{
	\begin{scope}[
		grow = right
	]
	\IfNoValueTF{#6}{
	 	\node[draw] (#1) at (#2,#3) {$#4$} 
		child {node[draw] {} edge from parent node[auto] {$#5$}}
		child[white] {node[black] {$\vdots$}};
	}{\IfNoValueTF{#7}{
		 	\node[draw] (#1) at (#2,#3) {$#4$} 
			child {node[draw] {} edge from parent node[auto] {$#5$}}
			child[white] {node[black] {$\vdots$}}
			child {node[draw] {} edge from parent node[auto] {$#6$}};
		}{
		 	\node[draw] (#1) at (#2,#3) {$#4$} 
			child {node[draw] {} edge from parent node[auto] {$#5$}}
			child[white] {node[black] {$\vdots$}}
			child {node[draw] {} edge from parent node[auto] {$#6$}}
			child {node[draw] {} edge from parent node[auto] {$#7$}};
		}
	}
	\end{scope}
}

\usetikzlibrary{trees} 

\NewDocumentCommand{\bouquet}{
s 
>{\SplitArgument{1}{,}}D(){0,0} 
m 
>{\SplitArgument{2}{,}}O{} 
}{
	\IfBooleanTF{#1}
		{\bAux#2{#3}#4}
		{\bAuxDots#2{#3}#4}
}

\NewDocumentCommand{\bAux}{
mm 
m 
mmm 
}{ 
	\draw (#1,#2) arc (-270:90:\defaultradius);
	\IfNoValueTF{#5}{
		\node[draw, fill = white] at (#1,#2) {$#3$}
		child {node[draw] {} edge from parent node[auto] {$#4$}};
	}{\IfNoValueTF{#6}{
			\node[draw, fill = white] at (#1,#2) {$#3$} [counterclockwise from = 0, sibling angle = 90]
			child {node[draw] {} edge from parent node[auto] {$#4$}}
			child {node[draw] {} edge from parent node[auto] {$#5$}};
		}{
			\node[draw, fill = white] at (#1,#2) {$#3$} [counterclockwise from = 0, sibling angle = 45]
			child {node[draw] {} edge from parent node[auto] {$#4$}}
			child {node[draw] {} edge from parent node[auto] {$#5$}}
			child {node[draw] {} edge from parent node[auto] {$#6$}};
		}
	}
}

\NewDocumentCommand{\bAuxDots}{
mm 
m 
mmm 
}{ 
	\draw (#1,#2) arc (-270:90:\defaultradius);
	\IfNoValueTF{#5}{
		\node[draw, fill = white] at (#1,#2) {$#3$} [counterclockwise from = 0, sibling angle = 30]
		child {node[draw] {} edge from parent node[auto] {$#4$}}
		child[white] {node[black] {$\ddots$}};
	}{\IfNoValueTF{#6}{
			\node[draw, fill = white]  at (#1,#2) {$#3$} [counterclockwise from = 0, sibling angle = 45]
			child {node[draw] {} edge from parent node[auto] {$#4$}}
			child[white] {node[black] {$\ddots$}}
			child {node[draw] {} edge from parent node[auto] {$#5$}};
		}{
			\node[draw, fill = white] at (#1,#2) {$#3$} [counterclockwise from = 0, sibling angle = 30]
			child {node[draw] {} edge from parent node[auto] {$#4$}}
			child[white] {node[black] {$\ddots$}}
			child {node[draw] {} edge from parent node[auto] {$#5$}}
			child {node[draw] {} edge from parent node[auto] {$#6$}};
		}
	}
}

\NewDocumentCommand{\bouquetdiagram}{
s 
D(){0,0} 
m 
O{} 
}{
\IfBooleanTF{#1}
{\diagram{\bouquet*(#2){#3}[#4]}}
{\diagram{\bouquet(#2){#3}[#4]}}
}

\NewDocumentCommand\eightdiagram{O{}} 
	{\diagram{\uloop(){#1} \dloop(){#1}}}

\NewDocumentCommand\elldiagram{O{}}
	{\diagram{\branch*(a)(0,-.5*\propagatorlength){#1} \node[draw](b) at (0,.5*\propagatorlength){}; \draw (a)--(b);}}

\NewDocumentCommand\eyediagram{O{}}
	{\diagram{\node[draw](a){$#1$}; \node[draw](b) at (1.5*\propagatorlength,0){}; \draw (a) to[bend right=45] (b); \draw (a) to[bend left=45] (b);}}

\NewDocumentCommand\tadpolediagram{O{}}
	{\diagram{\uloop(a){#1} \node[draw] (b) at (\propagatorlength,0) {}; \draw (a)--(b);}}

\author{
	Rodrigo Vargas Le-Bert\thanks{{\tt vargonis@gmail.com}. Supported by Associazione LumbeLumbe.}
}
\affil{School of Architecture and Design, University of Camerino}
\date{\normalsize\today}

\title{
Wick Ordering and 
Kinetic Energy Renormalization for Lévy White Noise Fields
}

\begin{document}
\maketitle

\begin{abstract}
Let $S=\T^d$ be a torus and $\mu$ the probability distribution of a Lévy white noise field $x:S\rightarrow\R$. Using projective limit measures we address the problem of making sense of $\e^{-T(x)}$, where $T(x) = \int_S \abs{\nabla x(s)}^2\D s$ is the kinetic energy, as a function in $L^1(\mu)$. We start by making sense of $T(x)$ itself as a sort of distribution, which is achieved by a generalization of Wick ordering. Then we specify to the case of a $\Gamma$ field, finding that Wick ordering does not eliminate all divergences. Higher order renormalization would be required, but the model seems to be non-renormalizable. 
(primary); 
60H40, 
82B30 
(secondary).
\end{abstract}

\setcounter{tocdepth}{2}
\tableofcontents

\section{Introduction}

Although
our conceptual understanding of renormalization in quantum and statistical physics has probably reached maturity with the advent of renormalization group methods, the same is certainly not true of our technical tools for dealing with it. From a mathematical perspective, where concepts and techniques are deeply interrelated by their joint evolution, this simply means that renormalization is badly understood. We believe, however, 
that the situation changes substantially once the construction of physically meaningful measures on spaces of fields is understood as a problem in cylinder measure theory. This is actually obvious once we recall that a cylinder measure is a compatible collection of measures on certain finite-dimensional quotients of a given space,  because the ultraviolet and infrarred cutoffs define a quotient of the space of fields; thus, a compatible family of effective theories is, indeed, a cylinder measure. But fully exploiting this connection requires relaxing a bit the standard notion of cylinder measure---something which from a purely mathematical perspective is well-motivated, too, as we proceed to argue.

\subsection{General preliminaries}

Cylinder measures~\cite{schwartz1973radon} are usually taken to be compatible families of measures on the finite-dimensional quotients of a topological vector space. This is not fully satisfactory for technical, practical and conceptual reasons which organize twofold. First, the linear structure is not essential; it is therefore obscuring the essence of the concept and hindering its use in cases where linearity is absent, or technically problematic. Second, in asking for a measure on \emph{all} possible finite-dimensional quotients, it becomes impractical. Indeed, as a general rule, calculations can only be done by commitment to \emph{particular} families of quotients, which become an auxiliary structure playing a role identical to that of coordinate systems in finite-dimensional geometry. Some measures will be tractable in one coordinate system, not in another. Moreover, and this is perhaps the most important point here, particular families of quotients will typically come with a scale parameter, letting the phenomenon of running constants characteristic of renormalization theory become apparent. 
Regarding coordinate systems,
one can use physical space coordinates; momentum space coordinates; and other coordinate systems that, by its own structure, a particular field might lend itself to, such as loop coordinates in gauge theory. 

Projective limit measures~\cite{yamasaki1985measures} address both shortcomings, but have in turn two limitations. First, full projective limits are typically too large. In general, the space of interest, say $X$, will be a subspace of a projective limit $\overline X$, with the projective system itself being auxiliary. Second, projective limits come with a natural topology, but that topology is almost invariably a bad choice for $X$. Thus, by a \emph{cylinder measure} on $X$ we will understand the given of a suitable embedding of $X$ into a projective limit $\overline X$ of finite dimensional spaces, together with a projective measure on $\overline X$; and we emphasize that we do not assume that $X$ inherits its topology from that of~$\overline X$.

Cylinder measures are an inevitable starting point in infinite-dim\-en\-sion\-al measure theory, because in practice one can only integrate functions which depend on a finite number of coordinates---i.e.\ \emph{cylinder functions.} But they are not the whole story, for in infinite dimensions coordinates do not fully specify a space: as already implied above, there is freedom for the topology, which will determine the regularity properties of the resulting measure (in general, cylinder measures are only finitely additive). Now, one might suppose that the space is fully given a priori and then ask under which conditions a cylinder measure on it will actually be Radon; or one might construct the cylinder measure first and look afterwards for a Radonifying operator providing for a satisfactory ambient space. 
This last approach is in line with the shift in focus from space to observables that has been so fruitful in physics and geometry, 
and we choose to adopt it. 
Following it implies some departure regarding the established methods in Euclidean field theory, as we proceed to describe.

Conventionally~\cite{simon2015p}, one starts with a free field represented by a Radon measure on the space of Schwartz distributions, and modifies it via a ``density'' inluding the interaction terms, which can itself be seen as a Hida distribution if one uses the machinery of white noise analysis~\cite{potthoff1993invariant}. The fact that incorporating the interaction requires renormalization shows that Radonifying the reference measure on Schwartz space has not helped: the main issue is still there, and 
it is related to the cylinder measure construction, not its Radonification. Thus, one can benefit from postponing Radonification, because pruning the unessential technicalities will make it easier to focus on the real issues and explore new avenues to solve them---which is what we set out to do here.

\subsection{Strategic overview}

Except for
postponing Radonification and working with a relaxed notion of cylinder measure, 
the strategy that we adopt here, and intend to follow in future work, does not depart from standard practice in general terms.
Getting into the details, though, in this paper we begin exploring an avenue that has been left untouched:
we take a Lévy white noise as a reference measure and attempt to add the kinetic energy as a perturbation. But let us not get ahead of ourselves.

The general problem is, formally, to make sense of a measure of the form
\[
\nu(\D x) = C\exp\left(-\int_S \mc L\bigl(x(s), \nabla x(s),\dots\bigr)\D s\right)\D x,
\]
where $S$ is space-time, $x$ is a field on $S$ (i.e.\ a section of some bundle, typically equipped with some geometric or algebraic structure), $\mc L$ is a Lagrangian density, $\D x$ is some sort of Lebesgue measure on the space of fields $X$, and $C$ is a normalizing constant. The strategy is: 
\begin{enumerate}
	\item Decompose $\mc L = \mc L_\mr{ref} + \mc L_\mr{pert}$ into a \emph{reference} and a \emph{perturbation} part.
	\item Choose a coordinate system: $X\subseteq \overline X = \projlim X^n$, in which the reference measure $\mu(\D x) = C\exp\left(-\int_S \mc L_\mr{ref}\bigl(x(s), \nabla x(s),\dots\bigr)\D s\right)\D x$ can be made sense of as a cylinder measure $\mu = \set{\mu^n}$.
	\item Construct a \emph{cylinder density} $f = \set{f^n}$, i.e.\ a family of functions such that $f\mu = \set{f^n\mu^n}$ is a cylinder measure, which, in a sense to be clarified, can be understood as being $f(x) = \exp\left(-\int_S \mc L_\mr{pert}\bigl(x(s),\nabla x(s),\dots\bigr)\D s\right)$.
	\item Radonify $f\mu$.
\end{enumerate}
Here we are concerned with the first three steps in the case of scalar fields, and we elaborate now on each of them. 

The first two steps are obviously closely interrelated. In the literature, the reference measure is choosen to be the free Gaussian measure, corresponding to $\mc L_\mr{ref}\bigl(x(s),\nabla x(s),\dots\bigr) = \Abs{\nabla x(s)}^2 + m^2\Abs{x(s)}^2$, where $m$ is a mass parameter. For the sake of simplicity, suppose that $S=\T^d$, the $d$-dimensional torus. The free measure diagonalizes in momentum space coordinates: $X\subseteq \projlim X^n$, where $X^n$ is the space of trigonometric polynomials of degree at most $n$. We will make the opposite choice: physical space coordinates, which diagonalize the potential energy part of the Lagrangian. Thus, it is the kinetic energy that will be the perturbation. This is a quite natural possibility when working with generalized cylinder measures. The construction of cylinder measures in physical coordinates leads immediately to the concepts of \emph{continuous product} measures and Lévy white noise. Running parameters already enter the scene, showing how renormalization lies behind the construction of a cylinder measure in almost all cases, except for the trivial instances of discrete products (such as the free field in momentum coordinates when space-time is compact). But truly heavy renormalization occurs only in the third step.

When attempting the construction of the cylinder density $f = \set{f^n}$, locality considerations must be taken into account. The problem is that $f^n$ should be the exponential of an effective perturbation Lagrangian $L_\mr{eff}^n$\,\Loo, in terms of which locality is the existence of a density:
\[
L_\mr{eff}^n(x) = \int_S \mc L_\mr{eff}^n\bigl(x(s),\nabla x(s),\dots\bigr)\D s.
\]
But the compatibility requirement for $\set{f^n}$ is incompatible with the locality requirement for $\Set{L_\mr{eff}^n}$, which must therefore be relaxed---as it should, because strict locality is actually a property that can only be expected to hold in the ultraviolet limit. Thus, we begin making sense of the third step by splitting it in two parts:
\begin{enumerate}[label = (3.\alph*)]
	\item
Choose approximate versions of the perturbation Lagrangian, i.e.\ functions $L_\mr{app}^n$ defined on $X^n$ such that 
\[
L_\mr{app}^n(x^n) \rightarrow \int_S \mc L_\mr{pert}\bigl(x(s),\nabla x(s),\dots\bigr)\D s,\quad x=\set{x^n}\in X\subseteq \projlim X^n.
\]
	\item
Find a family of Lagrangians $\set{L_\mr{ren}^n}$, minimally departing from $\set{L_\mr{app}^n}$ in a sense to be clarified, such that a cylinder density of the form
\[
f^n(x) = C^n\exp\left(-L_\mr{ren}^n\right) + o(1)
\]
can be constructed,
where $o(1)\rightarrow 0$ as $n\rightarrow\infty$ and $C^n$ is a normalizing constant.
\end{enumerate}
So, $L_\mr{ren}^n$ should be a constructed from $L_\mr{app}^n$\,\loo, ideally not being way too far from it. It will generically consist of $L_\mr{app}^n$ plus some counter-terms which diverge with the cutoff $n$. This is the renormalization problem for the perturbation Lagrangian, and attempting to solve it requires its own strategic choice. We proceed to state and justify ours, which, again, is not too different from established practice.

Step (a) is quite simple; the real action occurs in step (b). We approach it by noting that, at first order, we should have
\(
f^n \approx 1 - L_\mr{app}^n
\)
and, therefore, one can start by turning $L_\mr{app}$ itself into a cylinder density. Now, the compatibility requirement for a cylinder density reads
\[
L_\mr{ren}^n = \EE{L_\mr{ren}^{n+m}}[x^n],
\]
where the conditional expectation is taken with respect to the reference measure $\mu$.
This equation enables one to identify a good class of counterterms, inferred by computing conditional expectations of the approximate interaction Lagrangian. There will be freedom in the resulting (first-order) renormalized Lagrangian $L_\mr{ren}^n$\,\loo, but a canonical choice is obtained by simply deleting the divergent terms in $\lim_{m\rightarrow\infty} \EE{L_\mathrm{app}^{n+m}}[x^n]$. We refer to this as \emph{first order renormalization,} which in the Gaussian reference case is done by Wick ordering $L_\mr{app}^n$\,\loo.

Once the approximate perturbation Lagrangian has been put in (generalized) Wick order, one might be lucky and have that the limit
\[
f^n(x) = \lim_{m\rightarrow\infty} C^{n+m}\EE{\exp\left(-L_\mr{ren}^{n+m}\right)}[x^n]
\]
exists (which amounts to existence of the effective Lagrangian) and is integrable (which is related to stability). If that is the case, the problem has been solved. 
This turns out to work in simple cases, such as $\phi^4$ theory in $d=2$ space-time dimensions. In general, however, in order to have convergence as $m\rightarrow\infty$ one needs to pick the counterterms more carefully. That can be done by introducing a formal parameter $\lambda$ in front of the renormalized Lagrangian (which will be set to 1 in the end), letting the counterterms depend on that parameter and asking for the coefficients in a suitable series expansion for the effective Lagrangian to be finite. That would be the problem of higher-order renormalization, from the cylinder measure point of view.

\subsection{Final introductory remarks}

In this work we present some results stemming from pursuing the cylinder measure approach to constructive field theory in the case of scalar fields. This approach is still in an early stage of development; consequently, and by necessity, the style here is rather exploratory at times. We study one main example: the $\Gamma$ field, arriving at the conclusion that first order renormalization of the kinetic energy is not enough for this model, which is likely to be non-renormalizable. Higher-order renormalization is not attempted. 

Before getting into the matter, we summarize here what we believe are some strong points of our approach, together with a couple of brief comments meant to help put it in context. This is obviously not supposed to be exhaustive.

\begin{itemize}
	\item 
Perturbing the potential energy, as opposed to the kinetic one, is a natural, obvious choice for studying strong coupling regimes. That being said, it has come as a surprise to us the realization that the natural running constant that the kinetic energy picks up when regarded as a perturbation is asymptotically vanishing in the ultraviolet limit. Thus, the regimes that can be studied when perturbing a measure containing only the potential energy are a priori out of reach if one takes the complementary point of view of perturbing a measure containing only the kinetic energy---and conversely.
	\item 
By considering arbitrary continuous product reference measures we have arrived at a generalization of Wick ordering which throws a lot of light on that concept.  There are several other aspects of our approach that fit very nicely within the conceptual framework that physicists have developed, suggesting that it has the potential of greatly extending and clarifying the usual technical tools of perturbation theory.
	\item
There are certain parallels between what we do and white noise theory, with our cylinder densities playing the role of Hida distributions~\cite{hida2013white}. The two technical frameworks are, however, entirely different, with ours being more concrete and computation-oriented.
	\item 
Our approach to the effective Lagrangian should eventually lead to something related to the Meyer expansion in constructive renormalization~\cite{rivasseau1991from}, which differs from perturbative renormalization in that non-local effective Lagrangians are considered.
\end{itemize}

\section{Reference measures}

\subsection{Cylinder measures and densities}

Let $\mc P$ be a directed set, $\set{X_P|P\in\mc P}$ a projective system of topological spaces with projective limit $\overline X$ and canonical projections $\pi_P:\overline X\rightarrow X_P$\loo, and $X$ a subspace of $\overline X$ which is \emph{full,} in the sense that $\pi_P(X)=X_P$\loo. By a harmless abuse of notation, we will usually write $P$ instead of $\pi_P$\loo.
It will also be convenient to write $P$ for the projection $X\rightarrow X_P$\loo, and even for the projection $X_Q\rightarrow X_P$ when $Q\succcurlyeq P$ is understood from the context. 
We think of $\set{P:X\rightarrow X_P|P\in\mc P}$ as a coordinate system on $X$. 
Now, given $P\in\mc P$, consider the algebra
\[
\Cb{X;P} = \Set{\vphantom{\hat A} f\in \Cb{X} | f \text{ factors through } P:X\rightarrow X_P }.
\]
If $P\preccurlyeq Q$, there is a natural inclusion $
\Cb{X;P}\hookrightarrow \Cb{X;Q}$. The resulting directed system has an algebraic injective limit 
\[
\Cb{X;\mc P} = \injlim \Set{\vphantom{\hat A} \Cb{X;P} | P\in\mc P }.
\]
The elements of $\Cb{X;\mc P}$ are called \emph{cylinder functions} (for the coordinate system in use). 

\begin{defn}
Let $X$ be a Tychonoff space equipped with a coordinate system $\Set{P:X\rightarrow X_P|P\in\mc P}$.
A \emph{cylinder measure} on $X$ is a family of Radon measures $\set{\mu_P\text{ on } X_P}$ which is \emph{compatible,} in the sense that
\[
P_*\mu_Q = \mu_P\lo,\quad\text{ for all } Q\succcurlyeq P.
\]
We can also adopt a dual point of view and define a cylinder measure as a compatible family of positive linear functionals $\set{\rho_P:\Cb{X_P}\rightarrow\R}$---or, in other words, a positive linear functional on the injective limit $\Cb{X;\mc P}$. When the measure $\mu$ is clear from the context, we will sometimes write 
\[
\int_X f(x)\mu(\D x) = \rho_P(f_P) = \int_{X_P}f_P(x_P)\mu_P(\D x_P),\quad f\in\Cb{X;\mc P},
\]
where $P\in\mc P$ and $f_P\in\Cb{X_P}$ are such that $f=f_P\circ P$.
\end{defn}

We will assume that $\mc P\cong\Set{P^0,P^1,\dots}$ and write $X^n$ for $P^nX$, etcetera. Given a reference cylinder measure $\mu$, it is useful to regard $\set{x^n}_{n\in\N}$ as a stochastic process, which we call \emph{resolution process.} 
Now, given a family $\set{f^n:X^n\rightarrow\R}$, consider the compatibility, or martingale condition
\[
f^n = \EE{f^m}[x^n],
\]
where the conditional expectation is taken with respect to $\mu^m$.
If the $f$'s are uniformly bounded in the supremum norm, then $f\mu = (f^n\mu^n)$ is a finite, signed cylinder measure.
Observe that the compatibility condition greatly simplifies if $\mu$ is a product, i.e.\ the resolution process has independent increments. Families satisfying the compatibility condition will be called \emph{cylinder densities.}

\subsection{Momentum space coordinates}

We are interested in the case of $X$ being a space of real-valued functions over a space-time manifold $S$.
Two natural choices of coordinate system are: physical space coordinates; and, if $S$ is a homogeneous space, momentum space coordinates. Although we will exclusively use physical space coordinates, we take the opportunity to exemplify the  cylinder measure concept with the free field in momentum coordinates.

For the sake of simplicity, consider the compact space-time $S=\T^d$. Let
\begin{align*}
X^n &= \Set{\text{trigonometric polynomials of degree at most $n$}} \\
&= \bigoplus_{i\leq n} X_i\lo,\text{ the space of homogeneous polynomials of degree } i.
\end{align*}
The space of fields is some (Banach completion of a) subspace $X$ of $\varprojlim X^n$ (and we can't tell which one beforehand, because we want it to be a vector space where certain cylinder measure, to be specified, Radonifies). Those are the momentum space coordinates. This projective limit is simply a countable product, and therefore reference measures are very easy to come up with: just choose one measure for each coordinate and take their product. 
But then again, not any such product will be physically meaningful. Physically, this coordinate system is well-suited to treat the free scalar field, whose Lagrangian consists of the kinetic energy plus a mass term. Indeed, the Fourier transform diagonalizes the Lagrangian density 
\[
\mc L_\mr{free}(x) = \frac12 \bigl( mx(s)^2 + \abs{\nabla x(s)}^2 \bigr),
\]
and the resulting measure becomes a product when using the Fourier modes as coordinates. 

\subsection{Physical space coordinates} \label{continuous products}

One can take the interaction part of the Lagrangian, as opposed to its free part, as reference measure. This leads one to recover, in the case of homogeneous fields, the Lévy white noise measures. 
 
Assuming that the interaction part of the Lagrangian density does not involve field derivatives, it should diagonalize in physical space coordinates. Now, in this coordinate system the measure we are interested in ought to be a \emph{continuous} product~\cite{guichardet2006symmetric}---and, as we will shortly see, in this case a non-trival interdependence between the  measures in the corresponding family of finite-dimensional projections is inevitable. Concretely, this interdependence ex\-press\-es itself by the appearance of running parameters in the effective Lagrangians, with parameterized families belonging to a very limited class (consisting of logarithms of infinitely divisible distributions) actually giving rise to an associated ``product'' measure. Moreover, the generic case seems to be that the unavoidable parameters diverge in the ultraviolet limit, showing that the notion of a Lagrangian density has certain limitations in this context. 
\begin{rk}
If one is interested in an interation potential which is not the logarithm of an infinitely divisible distribution, then a part of it must inevitably go into the remainder Lagrangian density.
\end{rk}

Let $(S,\D s)$ be a measure space. 
We take a system $\set{p_i | i=1\dots n}$ of projections of the von Neumann algebra $L^\infty(S)$ which is orthogonal and complete, in the sense that $p_ip_j=0$ and $\sum p_i=1$. 
To $\{p_i\}$ we associate the conditional expectation
\[
P:L^\infty(S)\rightarrow X_P\loo,\quad
P = \sum_i p_ip_i^*\loo,\ p^*(x) = \int_S \bar px,
\]
where $\bar p = p/\abs p$ and $\abs p = \int_X p$. 
Now, let $\set{q_{ij} | i=1\dots n,\,j=1\dots m}$ be a refinement of $\{p_i\}$, i.e.\ another complete system of orthogonal projections such that $p_i=\sum_{j} q_{ij}$, with associated conditional expectation $Q:L^\infty(S)\rightarrow X_Q$\loo. Since $\{q_{ij}\}$ is a refinement of $\{p_i\}$, we have a projection (conditional expectation) $X_Q\rightarrow X_P$\loo. Given a directed family $\mc P$ of such systems of orthogonal projections we get a projective system $\{X_P\}$ and, if the family generates $L^\infty(S)$, then any good $X\subseteq L^\infty(S)$ will become a full subspace of $\projlim X_P$\loo. 

Let us construct a cylinder measure in this coordinate system. For the sake of definiteness, assume that $S$ is a $d$-dimensional manifold which is divided into hyper-cubed regions, with each $p_i$ corresponding to a hyper-cube (if there are infinitely many of them our cylinder measure will be a projective limit of projective limits, but that's not a problem), and that these regions are, in turn, subdivided into $r=2^d$ hyper-cubes, with projections $\{p_{ij}\}$, etcetera. Thus, $\abs{p_{i_1\cdots i_{n+1}}} = \abs{p_{i_1\cdots i_n}}/r$.
We will use the variables
\[
x_{i_1\cdots i_n} = p_{i_1\cdots i_n}^*(x),\quad x'_{i_1\cdots i_{n+1}} = \left( p_{i_1\cdots i_{n+1}}^*- p_{i_1\cdots i_n}^* \right)(x)
\]
so that 
\[
x_{i_1\cdots i_{n+1}} = \left\{\begin{aligned}
&x_{i_1\cdots i_n} + x'_{i_1\cdots i_n i_{n+1}}& &i_{n+1}<r,\\
&x_{i_1\cdots i_n} - \sum_{j=1}^{r-1}x'_{i_1\cdots i_n j}& &i_{n+1}=r.
\end{aligned}\right.
\]
Writing $\D x_{i_1\cdots i_n(\cdot)} = \D x_{i_1\cdots i_n,1}\cdots\D x_{i_1\cdots i_n,r}$ and $\D x'_{i_1\cdots i_n(\cdot)} = \D x'_{i_1\cdots i_n,1}\cdots \D x'_{i_1\cdots i_n,r-1}$ (exterior products are meant), one has that $\D x_{i_1\cdots i_n(\cdot)}$ equals
\begin{align*}
	&\bigl(\D x_{i_1\cdots i_n}+\D x'_{i_1\cdots i_n,1}\bigr)\cdots\bigl(\D x_{i_1\cdots i_n} + \D x'_{i_1\cdots i_n,r-1}\bigr)\left(\D x_{i_1\cdots i_n} - \sum_{j=1}^{r-1} \D x'_{i_1\cdots i_n,j} \right) \\
	&\quad = \sum_{j=1}^{r-1} \D x'_{i_1\cdots i_n,1}\cdots \underbrace{\D x_{i_1\cdots i_n}}_{j\text{-th place}} \cdots \D x'_{i_1\cdots i_n,r-1}(-\D x'_{i_1\cdots i_n,j}) + \D x'_{i_1\cdots i_n(\cdot)}\D x_{i_1\cdots i_n} \\
	&\quad = r\D x'_{i_1\cdots i_n(\cdot)}\D x_{i_1\cdots i_n}\LO.
\end{align*}
Now write $\D\mu^n(x^n) =  f^n(x^n)\D x^n = \prod f_{i_1\cdots i_n}(x_{i_1\cdots i_n})\D x_{i_1\cdots i_n}$\loo. In solving for $\left(f_{i_1\cdots i_n}\right)$ in terms of $\left(f_{i_1\cdots i_ni_{n+1}}\right)$, the partial integration factors into low-dimensional integrals involving only variables which are internal to each $i_1\cdots i_n$ region. Omitting these indices, one gets
\begin{align*}
f(x) &= r\idotsint\D x'_1\cdots\D x'_{r-1}\, f_1\bigl(x+x'_1\bigr)\cdots f_{r-1}\bigl(x+x'_{r-1}\bigr)f_r\bigl(x - (x'_1+\cdots+ x'_{r-1})\bigr) \\
	&= r\int\D x'_1\, f_1\bigl(x+x'_1\bigr)\cdots \int\D x'_{r-2}\, f_{r-2}\bigl(x+x'_{r-2}\bigr) \\
	&\qquad \cdot
\int\D x'_{r-1}\, f_{r-1}\bigl(x+x'_{r-1}\bigr) f_r\bigl((x - x'_1-\cdots-x'_{r-2}) - x'_{r-1})\bigr) \\
	&= r\int\D x'_1\, f_1\bigl(x+x'_1\bigr)\cdots \int\D x'_{r-2}\, f_{r-2}\bigl(x+x'_{r-2}\bigr) \\
	&\qquad \cdot (f_{r-1}*f_r)\bigl((2x - x'_1-\cdots-x'_{r-3})-x'_{r-2}\bigr) \\
	&\ \, \vdots\\
	&= r\bigl( f_1*\cdots*f_r \bigr)(rx).
\end{align*}
Solutions to these compatibility conditions can be obtained by writing them in terms of the Fourier transforms, for which---putting back the omitted indices---they read
\[ 
\hat f_{i_1\cdots i_n}(\xi) = \hat f_{i_1\cdots i_n,1}(\xi/r)\cdots \hat f_{i_1\cdots i_n,r}(\xi/r),\quad r=2^d.
\] 
For homogeneous solutions, the compatibility conditions boil down to
\[
\hat f_{n+1} = \hat f_n(r\xi)^{1/r}.
\]
Thus, any infinitely divisible distribution will provide a solution. 

\begin{ex}
Take $\hat f_n(\xi) = \e^{-r^n\xi^2/2}$.
This is a Gaussian measure with density
\[
f^n(x) = C\e^{-\frac1{r^n}\sum x_i^2}. 
\]
Assuming (without loss of generality) that the initial hyper-cubes $p_i$ have unit volume, then the hyper-cubes $p_{i_1\cdots i_n}$ have volume $r^{-n}$, the sum in the exponential is a Riemann sum and the resulting cylinder measure can be formally written
\[
\mu(\D x) = C\e^{-\frac12\int_S x(s)^2\D s}\D x = \Prodi_s C_s\e^{-\frac12 x(s)^2\D s}\D x(s),
\]
where $\prodi$ stands for a continuous product. 
\end{ex}
\begin{ex}
Let us see a case with running parameters: the Cauchy field, for which $\hat f_n(\xi) = \e^{-\abs{\xi}}$. The resulting measure has density
\[
f^n(x) = C\prod(1+x_i^2)^{-1}.
\]
This can be cast in the form $C\e^{-L^n(x)}$ with
\[
L^n(x) = 
 \sum\log( 1+x_i^2 ).
\]
This time, we don't get a Riemann sum. In order to obtain an expression analogous to that for the Gaussian above, we introduce the running parameter $\lambda(n) = r^n$, so that we can formally write
\[
\mu(\D x) = C\e^{-\lambda\int_S\log\bigl(1+x(s)^2\bigr)\D s}\D x,
\]
which is a strong coupling limit for the Lagrangian  $L(x) = \int_S\log\bigl(1+x(s)^2\bigr)\D s$.
\end{ex}
\begin{rk}
In the case $S=\R^d$,
the measures that we have just constructed can be easily obtained from the Bochner-Minlos theorem applied to the characteristic function 
\[
\EE{\e^{\I\langle\xi,x\rangle}} = \exp\left\{ \int_S c\bigl(\xi(s)\bigr)\D s \right\}, 
\]
where
\[
c(\xi) = \I b\xi - \frac12\sigma\xi^2 + \lambda\int (\e^{\I\xi x}-1)\D r(x) 
\]
is a Lévy characteristic. 
We will, however, rely on the concrete representation worked out above for making explicit calculations.
\end{rk}

\section{Wick ordering}

As explained in the introduction, first order renormalization 
gives rise to a notion of generalized Wick ordering, reproducing standard Wick ordering when the reference measure is Gaussian. Stating and solving the compatibility conditions for polynomial densities involves computing their conditional expectations, as we proceed to do now.

\subsection{Conditional expectation of polynomials}

Let $\mu=(\mu^n)$ be a product measure in physical space coordinates. We will suppose that $S$ has volume 1 and, as in \autoref{continuous products}, that each refinement $\set{p_{i_1\dots i_{n+1}}}$ of $\set{p_{i_1\dots i_n}}$ subdivides its hypercubes into $r=2^d$ hypercubes which, at step $n$, will have sidelength $\eps^n$ where $\eps=1/2$. Thus, each index $i_k$ runs over $I=\set{1,\dots,r}$.
Writing $x^n=(x_i)_{i\in I^n}$ and $x^{n+m} = (x_{ij})_{i\in I^n,j\in I^m}$,
we want to compute expectation values of the form
\begin{equation} \label{poly}
\EE{(x_{ij_1})^{k_1}\cdots (x_{ij_p})^{k_p}}[x^n]. 
\end{equation}
So, let $y^{n+m} = (y_{ij})_{i\in I^n,j\in I^m}$ be such that $\sum_j y_{ij}=0$ for each $i$ and
\(
x_{ij} = x_i + y_{ij}\lo.
\) 
Convening that $y^m_i = (y_{ij})_{j\in I^m}$ and $\D y_i^{m} = \prod_{j\neq (r,\dots, r)}\D y_{ij}$ one has that
\begin{align*}
\D\mu^{n+m}(x^{n+m}) &= \left(\prod_{i,j} f_{n+m}(x_{ij}) \right) \D x^{n+m} =  \left(\prod_{i,j} f_{n+m}(x_i+y_{ij}) \right)\left( \prod_i r^m\D x_i\D y_i^{m}\right) \\
	&= \prod_i \left(r^m \prod_j f_{n+m}(x_i+y_{ij}) \right)\D x_i\D y_i^{m} =: \prod_i f_n^m\bigl(x_i\lo,y^m_i\bigr) \D x_i\D y^m_i
\end{align*}
and, by the compatibility condition, 
\begin{align*}
f_n(x) &= \int f^m_n(x,y^m)\D y^m \\
	&= r^m\idotsint \D y_1\cdots\D y_{r^m-1} \prod_{j=1}^{r^m} f_{n+m}(x+y_j),\quad y_{r^m} = -\sum_{j=1}^{r^m-1}y_j\lo,
\end{align*}
where we have dropped the index $i$ and renamed the mute $y$-variables. 

\subsubsection{First-order renormalizability hypothesis}

We want expectation values as in~\eqref{poly} to be given by another polynomial of the same degree. This can be ensured under the following assumption.

\begin{hypo} \label{renormalizability}
For each $t\geq 0$ and $i\in\N$, there are constants $\Set{c_{ij}(t) | j\leq i+1}$ such that
\[
\hat f_t^{(i)}\hat f_t' = \sum_{j=0}^{i+1}c_{ij}(t) \hat f_t^{(j)}\hat f_t\lo,
\]
where $\set{f_t}_{t\geq 0}$ interpolates $\set{f_n}_{n\in\N}$ by $\hat f_t(\xi) = \hat f_0(r^t\xi)^{1/r^t}$.
\end{hypo}
\begin{prop}
Under \autoref{renormalizability}, for all $t\geq 0$, $k\in\N$ and $\lambda>0$, the vector spaces generated by the sets $\Set{\left(\hat f_t^\lambda\right)^{(\ell)} | \ell\leq k}$ and $\Set{\hat f_t^{(\ell)}\hat f_t^{\lambda-1} | \ell\leq k}$ are equal.
\end{prop}
\begin{proof}
It suffices to show that the vector space generated by $\Set{(\hat f_t^\lambda)^{(\ell)} | \ell\leq k}$ is contained in that generated by $\Set{\hat f_t^{(\ell)}\hat f_t^{\lambda-1} | \ell\leq k}$. Equality would then follow, because
\[
\hat f_t^\lambda(\xi) = \hat f_0(r^t\xi)^{\lambda/r^t} = \hat f_{t\log_r 1/\lambda}(\lambda\xi)
\]
and therefore
\begin{align*}
\Langle\Set{ \hat f_t^{(\ell)}|\ell\leq k }\Rangle
	&= \Langle \Set{ \left(\hat f_{t\log_r 1/\lambda}^{1/\lambda}\right)^{(\ell)} |\ell\leq k } \Rangle \\
	&\subseteq \Langle \Set{ \hat f_{t\log_r 1/\lambda}^{(\ell)}\hat f_{t\log_r 1/\lambda}^{1/\lambda-1} |\ell\leq k } \Rangle 
	= \Langle \Set{ (\hat f_t^\lambda)^{(\ell)}\hat f_t^{1-\lambda} |\ell\leq k } \Rangle.
\end{align*}
Now, we will do this by induction on $k$. The cases $k=0,1$ are trivial. Then, by inductive hypothesis,
\[ 
\left(\hat f_t^\lambda\right)^{(k+1)} 
	= \left(\sum_{i=1}^{k} c_i\hat f_t^{(i)}\hat f_t^{\lambda-1} \right)' = \sum_{i=1}^k c_i\hat f_t^{(i+1)}\hat f_t^{\lambda-1} + (\lambda-1)\sum_{i=1}^k c_i\hat f_t^{(i)}\hat f_t^{\lambda-2}\hat f_t', 
\] 
and we conclude by a direct application of \autoref{renormalizability}.  
\end{proof}

Let $f = f_0$ and denote by $R^k(\lambda)\in\mathrm{GL}(k+1,\C)$ the complex matrix given by the equation
\[
f^{\lambda-1}\left(\begin{matrix} \hat f \\ \hat f' \\ \vdots \\ \hat f^{(k)} \end{matrix}\right) = R^k(\lambda)\left(\begin{matrix} \hat f^\lambda \\ (\hat f^\lambda)' \\ \vdots \\ (\hat f^\lambda)^{(k)} \end{matrix}\right).
\]
In particular, $R^1(\lambda) = \begin{pmatrix} 1 &0 \\ 0 &1/\lambda \end{pmatrix}$, but further terms will depend on $f$.
Observe that we can safely drop the superscript $k$, because $R^k(\lambda)$ is a lower-triangular matrix obtained from $R^{k+1}(\lambda)$ by simply erasing the last line and column---and the same will apply to their inverses.
We also define $R(\mu,\lambda) = R(\mu)^{-1}R(\lambda)$, so that 
\[
\underbrace{\hat f^{1-\mu} \hat f^{\lambda-1}}_{\hat f^{\lambda-\mu}} \begin{pmatrix} \hat f^\mu \\ (\hat f^\mu)' \\ (\hat f^\mu)'' \\ \vdots  \end{pmatrix} = R(\mu,\lambda) \begin{pmatrix} \hat f^\lambda \\ (\hat f^\lambda)' \\ (\hat f^\lambda)'' \\ \vdots \end{pmatrix}
\]
and $R(\nu,\mu)R(\mu,\lambda) = R(\nu,\lambda)$.
\begin{lem} \label{cond_exp_lem}
Convening that the entries of $R$ are indexed starting from 0, one has that
\begin{align*}
&\int\D y_2\, f_{n+m}^{*p}(y_2) f_{n+m}^{*(r^m-p-q)}(r^mx-y_1-y_2)y_2^k \\
	&\qquad= \sum_{\ell=0}^{k} (-\I r^{n+m})^k (\I r^{-n})^\ell R_{k\ell}(p/r^{n+m}, (r^m-q)/r^{n+m})  \\
	&\pushright{ \cdot  \left(x-y_1/r^m\right)^{\ell} f_{n+m}^{*r^m-q}(r^mx-y_1). }
\end{align*}
In particular,
\begin{align*}
&r^m\int\D y\, f_{n+m}^{*p}(y)f_{n+m}^{*(r^m-p)}( r^mx-y ) y^k \\
	&\qquad= \sum_{\ell=0}^k (-\I r^{n+m})^k(\I r^{-n})^{\ell} R_{k\ell}(pr^{-n-m}, r^{-n})  x^\ell f_n(x)
\end{align*}
and $r^m\int\D y\, f_{n+m}^{*p}(y) f_{n+m}^{*(r^m-p)}(r^mx-y)y = pxf_n(x)$.
\end{lem}
\begin{proof}
Let $S=\left(\begin{smallmatrix} 1 & & & \\ &r & & \\ & &r^2 & \\ & & &\ddots \end{smallmatrix}\right)$, and note that
\[
\hat f_{n+m}^{r^m-q-p} \begin{pmatrix} \hat f_{n+m}^p \\ (\hat f_{n+m}^p)'\\ (\hat f_{n+m}^p)'' \\ \vdots \end{pmatrix} = S^{n+m} R(p/r^{n+m}, (r^m-q)/r^{n+m}) S^{-(n+m)} \begin{pmatrix} \hat f_{n+m}^{r^m-q} \\ (\hat f_{n+m}^{r^m-q})' \\ (\hat f_{n+m}^{r^m-q})'' \\ \vdots  \end{pmatrix}
\]
which follows from
\[
\hat f^{1/r^n - (p+q)/r^{n+m}} \begin{pmatrix} \hat f^{p/r^{n+m}} \\ (\hat f^{p/r^{n+m}})' \\ \vdots  \end{pmatrix} = R(p/r^{n+m},(r^m-q)/r^{n+m}) \begin{pmatrix} \hat f^{1/r^n-q/r^{n+m}} \\ (\hat f^{1/r^n-q/r^{n+m}})' \\ \vdots  \end{pmatrix}
\]
by noting that
\[
(\hat f_{n+m}^p)^{(k)} = \left(\hat f^{p/r^{n+m}}(r^{n+m}\cdot)\right)^{(k)} = r^{(n+m)k}\left(\hat f^{p/r^{n+m}}\right)^{(k)}(r^{n+m}\cdot).
\]
Thus,
\begin{align*}
&\int\D y_2\, f_{n+m}^{*p}(y_2)f_{n+m}^{*(r^m-q-p)}( r^mx-y_1-y_2 ) y_2^k \\
	&\qquad= \left( (-\I)^k (\hat f_{n+m}^p)^{(k)}\hat f_{n+m}^{r^m-q-p} \right)^\vee(r^mx-y_1) \\
	&\qquad= (-\I)^k r^{(n+m)k} \sum_{\ell=0}^k R_{k\ell}(p/r^{n+m},(r^m-q)/r^{n+m}) r^{-(n+m)\ell} \\
	&\pushright{ \cdot \bigl(\I(r^mx-y_1)\bigr)^\ell f_{n+m}^{*r^m-q}(r^mx-y_1) } \\
	&\qquad= \sum_{\ell'\leq\ell\leq k} {\ell\choose\ell'} (-\I r^{n+m})^k (\I r^{-n})^\ell R_{k\ell}(p/r^{n+m}, (r^m-q)/r^{n+m})  \\
	&\pushright{ \cdot x^{\ell-\ell'}(-r^{-m}y_1)^{\ell'} f_{n+m}^{*r^m-q}(r^mx-y_1), }
\end{align*}
as claimed.
\end{proof}

\begin{theo} \label{cond_exp_of_powers}
\[
\EE{(x_{ij})^k}[x^n] = \sum_{\ell=0}^k (-\I r^{n+m})^k (\I r^{-n})^{\ell} R_{k\ell}(r^{-n-m},r^{-n}) x_i^{\ell}\LO.
\]
In particular, $\EE{x_{ij}}[x^n] = x_i$ and $\EE{y_{ij}}[x^n]=0$.
\end{theo}
\begin{proof}
Without loss of generality, suppose that $j=1$ and drop the $i$ indices. Applying \autoref{cond_exp_lem} we compute
\begin{align*}
&\frac1{f_n(x)} \int\D y^m\, f_n^m(x,y^m) (x+y_1)^k \\
	&\qquad = \frac{r^m}{f_n(x)} \int\D y_1\, f_{n+m}(x+y_1) f_{n+m}^{*(r^m-1)}\bigl((r^m-1)x - y_1\bigr) (x+y_1)^k \\
	&\qquad = \frac{r^m}{f_n(x)} \int\D y_1\, f_{n+m}(y_1)f_{n+m}^{*(r^m-1)}\bigl( r^mx-y_1 \bigr) y_1^k \\
	&\qquad =  (-\I r^{n+m})^k \sum_{\ell=0}^k R_{k\ell}(r^{-n-m}, r^{-n}) (\I r^{-n}x)^{\ell}. \qedhere
\end{align*}
\end{proof}
Of course, analogous expressions can be obtained for things like~\eqref{poly}, but we won't write them explicitely except for one simple example: the $\Gamma$ reference measure.

\subsubsection{The $\Gamma$ field}

Take $\hat f_n(\xi) = (1-\I\xi/\beta_n)^{-\alpha_n}$, with $\alpha_n=\alpha_0/r^n$ and $\beta_n=\beta_0/r^n$. In the next few paragraphs we will omit the subindex $n$ for convenience. Since
\begin{align*}
\hat f^{(k)}\hat f' &= \I^k\alpha(\alpha+1)\cdots(\alpha+k-1)/\beta^k\cdot (1-\I\xi/\beta)^{-\alpha-k} \cdot \I\alpha/\beta \cdot(1-\I\xi/\beta)^{-\alpha-1} \\
	&= \frac{\I^{k+1}\alpha^2(\alpha+1)\cdots(\alpha+k-1)/\beta^{k+1}}{\I^{k+1}\alpha(\alpha+1)\cdots(\alpha+k)/\beta^{k+1}} \hat f^{(k+1)}\hat f 
	= \frac{\alpha}{\alpha+k} \hat f^{(k+1)}\hat f,
\end{align*}
\autoref{renormalizability} is satisfied. 
Now,
\begin{align*}
(\hat f^\lambda)^{(k)} &= \I^k\lambda\alpha(\lambda\alpha+1)\cdots(\lambda\alpha+k-1)/\beta^k\cdot (1-\I\xi/\beta)^{-\lambda\alpha-k} \\
	&= \frac{\lambda\alpha(\lambda\alpha+1)\cdots(\lambda\alpha+k-1)}{\alpha(\alpha+1)\cdots(\alpha+k-1)} \hat f^{\lambda-1}\hat f^{(k)} = \frac{(\lambda\alpha)_{k}}{(\alpha)_{k}} \hat f^{\lambda-1}\hat f^{(k)}
\end{align*}
where $(x)_n = x(x+1)\cdots (x+n-1)$ is the Pochhammer symbol, and
\[
R(\lambda)^{-1} = \left(\begin{matrix} 
1 & & &  \\ 
  &\lambda & &  \\
  & &\frac{(\lambda\alpha)_{2}}{(\alpha)_2} &  \\
  & & & \ddots  \\
\end{matrix}\right).
\]
Therefore,
\[
\EE{(x_{ij})^k}[x^n] 
	= r^{mk} R_{kk}(r^{-n-m}, r^{-n}) x_i^k 
	= r^{mk} \frac{(\alpha_{n+m})_{k}}{(\alpha_n)_{k}} x_i^k\Lo.
\]
More generally, we have the following result.
\begin{theo} \label{expectation of monomials}
\[
\EE{(x_{ij_1})^{k_1}\cdots (x_{ij_p})^{k_p}}[x^n] 
	= r^{mk}\frac{ (\alpha_{n+m})_{k_1}\cdots (\alpha_{n+m})_{k_p}}{ (\alpha_n)_{k} }x_i^k\Lo,\quad
	k=k_1+\cdots+k_p\lo.
\]
\end{theo}
\begin{proof}
By repeated application of \autoref{cond_exp_lem},
\begingroup
\allowdisplaybreaks
\begin{align*}
&\EE{(x_{ij_1})^{k_1}\cdots (x_{ij_p})^{k_p}}[x^n] \\
	&\quad= \frac1{f_n(x)} \idotsint\D y_1\cdots\D y_p\, f_{n+m}^{*(r^m-p)}\biggl( (r^m-p)x-\sum_{\ell=1}^p y_\ell \biggr) \prod_{\ell=1}^p f_{n+m}(x+y_\ell)(x+y_\ell)^{k_\ell} \\
	&\quad= \frac1{f_n(x)} \idotsint\D y_1\cdots\D y_p\, f_{n+m}^{*(r^m-p)}\biggl( r^mx-\sum_{\ell=1}^p y_\ell \biggr) \prod_{\ell=1}^p f_{n+m}(y_\ell)y_\ell^{k_\ell} \\
	&\quad= \frac1{f_n(x)} \idotsint\D y_2\cdots \D y_{p}\, R_{k_1k_1}\bigl(r^{-n-m},r^{-n}-(p-1)r^{-n-m}\bigr) \biggl( r^mx - \sum_{\ell=2}^{p}y_\ell \biggr)^{k_1} \\
	&\pushright{ {}\cdot f_{n+m}^{*(r^m-p-1)}\biggl( r^mx-\sum_{\ell=2}^{p} y_\ell \biggr) \prod_{\ell=2}^{p} f_{n+m}(y_\ell)y_\ell^{k_\ell} } \\
	&\quad= \frac{1}{f_n(x)}\idotsint\D y_2\cdots\D y_p\, C_1\sum_{j=0}^{k_1} (-1)^j {k_1\choose j} \biggl( r^mx - \sum_{\ell=3}^p y_\ell \biggr)^{k_1-j}   \\
	&\pushright{ {}\cdot f_{n+m}^{*(r^m-p-1)}\biggl( r^mx-\sum_{\ell=2}^p y_\ell \biggr) f_{n+m}(y_2) y_2^{k_2+j} \prod_{\ell=3}^p f_{n+m}(y_\ell) y_\ell^{k_\ell} } \\
	&\quad= \frac{1}{f_n(x)} C_1\sum_{j=0}^{k_1} (-1)^j C_{2j} {k_1\choose j} \idotsint\D y_3\cdots\D y_p \biggl(r^mx - \sum_{\ell=3}^p y_\ell\biggr)^{k_1-j} \biggl(r^mx - \sum_{\ell=3}^p y_\ell\biggr)^{k_2+j} \\
	&\pushright{ {}\cdot f_{n+m}^{*(r^m-p-2)} \biggl(r^mx - \sum_{\ell=3}^p y_\ell\biggr) \prod_{\ell=3}^p f_{n+m}(y_\ell)y_\ell^{k_\ell} } \\
	&\quad\ \, \vdots\\
	&\quad= \frac{1}{f_n(x)} C_1C_2\cdots C_p (r^mx)^{k_1+\cdots+k_p} f_{n+m}^{*r^m}(r^m x) 
	= x^{k} r^{mk} C_1C_2\cdots C_p
\end{align*}%
\endgroup
where $k=k_1+\cdots+k_p$ and 
\begin{align*}
C_\ell 
	&= \sum_{j=0}^{k_1+\cdots +k_{\ell-1}} (-1)^j {k_1+\cdots+k_{\ell-1} \choose j} R_{k_\ell+j, k_\ell+j}\bigl( r^{-n-m}, r^{-n}-(p-\ell)r^{-n-m} \bigr) \\
	&= \sum_{j=0}^{k_1+\cdots+k_{\ell-1}} (-1)^j {k_1+\cdots+k_{\ell-1} \choose j} \frac{(\alpha_{n+m})_{k_\ell+j}}{\bigl(\alpha_n-(p-\ell)\alpha_{n+m}\bigr)_{k_\ell+j}}. 
\end{align*}
Now,
\begin{align*}
C_2 &= \sum_{j=0}^{k_1} (-1)^j{k_1\choose j} \frac{(\alpha_{n+m})_{k_2+j}}{ \bigl(\alpha_n - (p-2)\alpha_{n+m}\bigr)_{k_2+j} } \\
	&= \sum_{j=0}^{k_1} (-1)^{k_1} {k_1\choose j} \frac{(\alpha_{n+m})_{k_2} (\alpha_{n+m}+k_2)_{j} \bigl( -\alpha_n + (p-2)\alpha_{n+m} - (k_1+k_2-1) \bigr)_{k_1-j}}{ \bigl(\alpha_n - (p-2)\alpha_{n+m}\bigr)_{k_1+k_2} } \\
	&= (-1)^{k_1} \frac{ (\alpha_{n+m})_{k_2} }{ \bigl(\alpha_n - (p-2)\alpha_{n+m} \bigr)_{k_1+k_2} } \bigl(-\alpha_n + (p-1)\alpha_{n+m} - (k_1-1)\bigr)_{k_1}
\end{align*}
so that
\[
C_1C_2 = \frac{ (\alpha_{n+m})_{k_1} (\alpha_{n+m})_{k_2}}{\bigl( \alpha_n - (p-2)\alpha_{n+m} \bigr)_{k_1+k_2}}.
\]
Inductively,
\[
C_1\cdots C_p = \frac{ (\alpha_{n+m})_{k_1}\cdots (\alpha_{n+m})_{k_p} }{ (\alpha_n)_{k} }. \qedhere
\]
\end{proof}

\subsection{Wick ordering polynomial potentials}

For a discrete product reference measure, such as the free field on the torus, the resolution process has independent increments and one can use Wick (or Appell) polynomials---which in that case are known to have the required martingale property---to renormalize the Lagrangian density. For continuous product measures, however, increments are not independent anymore and we have to work out the right generalization. Now, we will not give a formal definition, but just note that there are two requirements:
\begin{itemize}
	\item
For each $k$, we want a family $\Set{\mc V^n_k(x)}$ of polynomials of degree $k$ such that $\Set{\frac1{r^n}\sum \mc V^n_k(x_i)}$ is a cylinder density.
	\item
The polynomials above must be somehow uniquely determined.\footnote{In the discrete product case, one can take them to be monic and orthogonal in $L^2(\mu^n)$ for $n$ large enough, but in the continuous product case multiplicative renormalization might be needed (cannot choose them monic anymore) and orthogonality becomes $n$ dependent.} 
\end{itemize}

With \autoref{cond_exp_of_powers} one is ready to renormalize polynomial interactions. As a trivial example, we have a candidate for renormalized polynomial of degree one: $\mc V^n_1(x)=x$. 
Indeed, 
\[
\frac1{r^{n+m}}\sum \EE{x_{ij}}[x^n] = \frac1{r^n} \sum x_i\loo,
\]
i.e.\ $\frac1{r^n}\sum\mc V^n_1(x_i)$ is a cylinder distribution. Will compute higher degree cases for some particular examples. Will adhere to the recipe of finding $\mc V^n_k$ by simply dropping the divergent terms in
\(
\lim_{m\rightarrow\infty} \EE{(x_{ij})^k}[x^n].
\)

\subsubsection{The $\Gamma$ field}

For renormalized powers one can take
\[
\mc V^n_k(x) = \frac{C}{(\alpha+r^n)\cdots(\alpha+(k-1)r^n)} x^k
	= \frac{C}{r^{kn}(\alpha_n)_k} x^k
\]
with the normalization $C=(\alpha_0)_k$ making $\mc V^0_k(x) = x^k$.
\begin{rk}
In stark contrast with the Gaussian case, there is multiplicative renormalization, and only that (no counterterms needed).
\end{rk}

\subsubsection{The Gaussian field}

Here, $\hat f_n(\xi) = \e^{-\sigma_n\xi^2/2}$ with $\sigma_n=r^n\sigma_0$\lo. Suppose for simplicity that $\sigma_0=1$, so that $\hat f_0^{(k)} = (-1)^kH^k(\xi)\hat f_0$ where $H^k(\xi)$ is the $k$-th Hermite polynomial. \autoref{renormalizability} is satisfied thanks to the recursion relation
\[
\xi H^k(\xi) = kH^{k-1}(\xi) + H^{k+1}(\xi).
\]
Let $f=f_0$\lo.
Since $\hat f^\lambda = \hat f(\lambda^{1/2}\cdot)$, one has
\begin{align*}
(\hat f^\lambda)^{(k)} &= \lambda^{k/2}\hat f^{(k)}(\lambda^{1/2}\cdot) = \lambda^{k/2} \bigl((-1)^k H^k(\lambda^{1/2}\cdot) \hat f\bigr)\hat f^{\lambda-1} \\
	&= (-1)^k\lambda^{k/2}\hat f^{\lambda-1} \sum_{i=0}^{\lfloor k/2\rfloor} \lambda^{k/2-i}(\lambda-1)^i{k\choose 2i}\frac{(2i)!}{2^{i}i!} H^{k-2i}\hat f \\
	&= \lambda^k\hat f^{\lambda-1}\sum_{i=0}^{\lfloor k/2\rfloor} \left(1-\lambda^{-1}\right)^i{k\choose 2i}\frac{(2i)!}{2^{i}i!}  \hat f^{(k-2i)}
\end{align*}
thanks to the multiplication theorem---which, in terms of $R$, takes the form
\[
\begin{pmatrix} 1 \\ H^1(\lambda^{1/2}\cdot) \\ H^2(\lambda^{1/2}\cdot) \\ \vdots \end{pmatrix} = \lambda^{k/2}R(\lambda)^{-1} \begin{pmatrix} 1 \\ H^1 \\ H^2 \\ \vdots \end{pmatrix}.
\]
Let us work out $\mc V_2, \mc V_3$ and $\mc V_4$. Explicitely, 
\[ 
R^4(\lambda)^{-1} 
	= \begin{pmatrix} 
1 & & & & \\
0 &\lambda & & & \\
\lambda^2(1-\lambda^{-1}) &0 &\lambda^2 & & \\
0 &3\lambda^3(1-\lambda^{-1}) &0 &\lambda^3 & \\
3\lambda^{4}(1-\lambda^{-1})^2 &0 &6\lambda^4(1-\lambda^{-1}) &0 &\lambda^4 \\
		\end{pmatrix}
\] 
and one can also check that $R^4(r^{-n-m},r^{-n})$ equals
\[
\left(\begin{smallmatrix}
1 & & & & \\
0 &r^{-m} & & & \\
r^{-(n+m)}(1-r^{-m}) &0 &r^{-2m} & & \\
0 &-3r^{-(n+m)}r^{-m}(1-r^{-m}) &0 &r^{-3m} & \\
3r^{-2(n+m)}(1-r^{-m})^2 &0 &-6r^{-(n+m)}r^{-2m}(1-r^{-m}) &0 &r^{-4m}
\end{smallmatrix}\right).
\]
Thus,
\begin{align*}
\EE{(x_{ij})^2}[x^n] &= (-\I r^{n+m})^2 \left( -r^{-(n+m)}(1-r^{-m}) + r^{-2m}(\I r^{-n})^2x_i^2 \right) \\
	&= x_i^2 + r^{n+m}(1-r^{-m}) \\
\EE{(x_{ij})^3}[x^n] &= (-\I r^{n+m})^3 \left( -3r^{-(n+m)}r^{-m}(1-r^{-m})(\I r^{-n})x_i + r^{-3m}(\I r^{-n})^3x_i^3 \right) \\
	&= x_i^3 + 3r^{n+m}(1-r^{-m})x_i \\
\EE{(x_{ij})^4}[x^n] &= (-\I r^{n+m})^4 \left( 3r^{-2(n+m)}(1-r^{-m})^2  \right. \\ &\pushright{ \left.{} - 6r^{-(n+m)}r^{-2m}(1-r^{-m})(\I r^{-n})^2x_i^2 + r^{-4m}(\I r^{-n})^4x_i^4 \right) } \\
	&= x_i^4 + 6r^{n+m}(1-r^{-m})x_i^2 + 3\left(r^{n+m}(1-r^{-m})\right)^2 
\end{align*}
and we see that we can take
\begin{align*}
\mc V^n_2(x) &= \lim_{m\rightarrow\infty} \EE{(x_{ij})^2-r^{n+m}}[x^n] 
	= x_i^2 - r^n \\
\mc V^n_3(x) &= \lim_{m\rightarrow\infty} \EE{(x_{ij})^3 - 3r^{n+m}x_{ij}}[x^n] 
	= x_i^3-3r^nx_i \\
\mc V^n_4(x) &= \lim_{m\rightarrow\infty} \EE{(x_{ij})^4-6r^{n+m}(x_{ij})^2+3r^{2(n+m)}}[x^n] \\
	&= x_i^4 - 6r^nx_i^2 + 3r^{2n} + 6r^{n+m}x_i^2 + 3r^{2(n+m)} - 6r^{n+m}r^n \\
	&\pushright{ {} - 6r^{n+m}x_i^2 - 6r^{2(n+m)} + 6r^{n+m}r^n + 3r^{2(n+m)} } \\
	&= x_i^4 - 6r^nx_i^2 + 3r^{2n}
\end{align*}
just as in the discrete case.
Observe how the final result is obtained by just dropping the divergent part in $\EE{(x_{ij})^k}[x^n]$.

\subsubsection{Further examples}

The Poisson and Meixner distributions also satisfy \autoref{renormalizability}.
As for infinitely divisible distributions which do not, we mention: Cauchy, Sym $\Gamma$, compound Poisson with exponentially distributed summands, and compound Poisson with Gaussian summands. 
Also, if one adds a diffusion term to the Lévy characteristic of a $\Gamma$, Poisson or Meixner distribution, the result stops satisfying \autoref{renormalizability}. 

\subsection{Wick ordering the kinetic energy} \label{kinetic energy renormalization}

\subsubsection{Conditional expectation of the approximate kinetic energy}

Consider the lattice approximation to the kinetic energy
\[ 
T_\text{app}^n(x^n) = \frac1{r^n} \sum_i\sum_{\abs{i'-i}=1} \frac12 \left( \frac{x_{i'}-x_i}{\eps^n} \right)^2 = 
\frac1{r^n} \sum_{i,i'} \eps^{-2n} c_{ii'}x_ix_{i'}
\] 
where $c_{ii'}$ equals: $d$ if $i=i'$; $-1$ if $i$ and $i'$ are nearest neighbours; and 0 otherwise. We will assume that $S$ is a compact manifold obtained from $[0,1]^d$ by some glueing of the boundary, so that counting neighbours can be easily done: every site has exactly $2d$ neighbours.

\begin{theo}
\[
\EE{T^{n+m}_\mathrm{app}}[x^n] =
	\eps^{-n-m}\left( \eps^nT_\mathrm{app}^n(x^n) + d\eps^{-n} T_\mathrm{shift}^{n,m}(x^n) \right)
\]
where
\(
T_\mathrm{shift}^n(x^n) = \frac{1}{r^n}\sum \mc T_\mathrm{shift}^{n,m}(x_i)
\)
and
\[
\mc T_\mathrm{shift}^{n,m}(x_i) = \eps^{-m} \EE{y_{ij}^2}[x^n] - (\eps^{-m}-1)\EE{y_{ij}y_{ij'}}[x^n].
\]
\end{theo}
\begin{proof}
Write
\begin{align*}
\mathbb E\bigl[T^{n+m}_\text{app} \bigm| x^n\bigr] 
	&= \frac1{r^{n+m}} \sum_{i,j,i',j'} \eps^{-2(n+m)} c_{iji'j'} \EE{(x_i+y_{ij})(x_{i'}+y_{i'j'})}[x^n] \\
	&=: \frac1{r^n}\sum_{i,i'} \eps^{-2n} \tilde c_{ii'}^m(x_i\lo,x_{i'})
\end{align*}
where $c_{iji'j'}$ equals: $d$ if $(i,j)=(i',j')$; $-1$ if $(i,j)$ and $(i',j')$ are nearest neighbours; and 0 otherwise. Thus,
\begin{align*}
\tilde c_{ii'}^m(x_i\lo,x_{i'}) 
	&= \frac1{r^m} \sum_{j,j'} \eps^{-2m} c_{iji'j'} \EE{x_ix_{i'} + (x_iy_{i'j'}+x_{i'}y_{ij}) + y_{ij}y_{i'j'}}[x^n] \\
	&= \frac1{r^m} \sum_{j,j'} \eps^{-2m} c_{iji'j'} \left(x_ix_{i'} + \EE{y_{ij}y_{i'j'}}[x^n] \right),
\end{align*}
thanks to \autoref{cond_exp_of_powers}.
Let's add up the terms not involving conditional expectations. If $i=i'$, get
\begin{align*}
&\frac{\eps^{-2m}}{r^m} x_i^2\biggl\{ dr^m - \frac12 \sum_{k=0}^{d} (d+k)\cdot\underbrace{\#(\text{hypercubes with $d+k$ neighbours})}_{{d\choose k} 2^{d-k}(\eps^{-m}-2)^k} \biggr\} \\
	&\qquad= \frac{\eps^{-2m}}{r^m} x_i^2 \biggl\{ dr^m - \frac12\biggl( d\eps^{-md} + \sum_{k=0}^d k{d\choose k} 2^{d-k} (\eps^{-m}-2)^k \biggr) \biggr\} \\
	&\qquad= \frac{\eps^{-2m}}{r^m} x_i^2 \biggr\{ dr^m - \frac12\left( d\eps^{-md} + d\eps^{-md} - 2d\eps^{-m(d-1)} \right) \biggr\} = \eps^{-m}dx_i^2\Lo,
\end{align*}
whereas if $i\neq i'$, get
\[
-\frac{\eps^{-2m}}{r^m}x_ix_{i'} \eps^{-m(d-1)} = -\eps^{-m} x_ix_{i'}\loo.
\]
Then, note that $\EE{y_{ij}y_{i'j'}}[x^n]$ vanishes if $i\neq i'$, so 
\begin{align*}
&\EE{T^{n+m}_\text{app}}[x^n] = \eps^{-m}T_\text{app}^n  \\
	&\pushright{{} + \frac{\eps^{-2n}}{r^n} \sum_{i} \frac{\eps^{-2m}}{r^m} \left\{ dr^m \EE{y_{ij}^2}[x^n] - (d\eps^{-md} - d\eps^{-m(d-1)}) \EE{y_{ij}y_{ij'}}[x^n] \right\} } \\
	&\qquad= \eps^{-n-m}\left( \eps^nT_\text{app}^n + \eps^{-n} \frac{d}{r^n}\sum \left\{ \eps^{-m} \EE{y_{ij}^2}[x^n] - (\eps^{-m}-1)\EE{y_{ij}y_{ij'}}[x^n] \right\} \right),
\end{align*}
as claimed.
\end{proof}
\begin{rk}
Let us pause for a moment to reflect on the fact that $\EE{T_\text{app}^{n+m}}[x^n]$ equals $\eps^{-m} T_\text{app}^n$ plus some (divergent) mass corrections. We see, hence, that in renormalizing the kinetic energy both additive and \emph{multiplicative} renormalization will be required---which is the reason why we explicitely factored out the $\eps^{-n-m}$ in $\EE{T_\text{app}^{n+m}}[x^n]$. Thus, we expect that there exists some cylinder density $T_\text{ren}$ such that
\[
T_\text{app}^n = \eps^{-n}T_\text{ren}^n + O(\eps^{-2n}).
\]
This indicates that if we simply make $n\rightarrow\infty$ in 
\(
C_n\e^{-T_\text{app}^n(x)} \mu^n(\D x),
\)
we would be studying some sort of ``weak coupling'' limit. 
But that would not be the most natural thing to do in this context, for clearly $T_\text{ren}$ is the natural perturbation of a continuous product measure. This strongly reinforces the idea that perturbing a continuous product measure (in physical space coordinates) is a promising strategy for the study of field theories in strong coupling regimes.
\end{rk}

Let us compute 
\begin{align*}
&\EE{y_{ij}y_{ij'}}[x^n] 
	= -x_i^2 + \EE{x_{ij}x_{ij'}}[x^n] \\
	&\qquad = -x_i^2 + \frac{r^m}{f_n(x_i)} \iint\D y_1\D y_2\, f_{n+m}(y_1)f_{n+m}(y_2)f_{n+m}^{*r^m-2}(r^mx_i-y_1-y_2)y_1y_2 \\
	&\qquad = -x_i^2 + \frac{r^m}{f_n(x_i)} \int\D y_1\frac{r^m}{r^m-1} f_{n+m}(y_1) f_{n+m}^{*(r^m-1)}(r^mx-y_1)(x_i-y_1/r^m)y_1 \\
	&\qquad= \frac1{r^m-1}x_i^2 - \frac{1}{r^m-1}\sum_{\ell=0}^2 (-\I r^{n+m})^2 R_{2,\ell}(r^{-n-m},r^{-n}) (\I r^{-n}x_i)^\ell \\
	&\qquad= \frac1{r^m-1}\left( x_i^2 - \EE{x_{ij}^2}[x^n] \right).
\end{align*}
So, noting that $\EE{y_{ij}^2}[x^n] = \EE{x_{ij}^2}[x^n] - x_i^2$\Lo, we can also write 
\begin{align*}
\mc T^{n,m}_\mr{shift}(x_i) &= \eps^{-m} \frac{r^m}{r^m-1} \left( \EE{x_{ij}^2}[x^n] - x_i^2 \right) - \frac1{r^m-1}\left( \EE{x_{ij}^2}[x^n] - x_i^2 \right) \\
	&= \frac{\eps^{-m}r^m-1}{r^m-1} \left( \EE{x_{ij}^2}[x^n] - x_i^2 \right).
\end{align*}
At any rate, $\mc T_\text{shift}^{n,m}$ is an (at most) quadratic polynomial. Thereby, quadratic counterterms will be needed to renormalize the kinetic energy. Let us see a couple of examples.

\subsubsection{Some examples}

We start with the $\Gamma$ field. We have
\[ 
\mc T_\text{shift}^{n,m}(x_i) 
	= \frac{\eps^{-m}r^m-1}{r^m-1} \left(\frac{r^m(\alpha_{n+m}+1)}{\alpha_n+1} - 1 \right) x_i^2 
	=  \frac{\eps^{-m}r^m-1}{\alpha_n+1} x_i^2
\] 
so that we can define
\begin{align*}
T_\text{ren}^n &= \lim_{m\rightarrow\infty} \EE{\eps^{n+m}T_\text{app}^{n+m} - \frac{d\eps^{-(n+m)}}{\alpha_{n+m}+1}\frac{1}{r^{n+m}} \sum_{i,j} x_{ij}^2}[x^n] \\
	&= \lim_{m\rightarrow\infty} \left\{  \eps^nT_\text{app}^n - \frac{d\eps^{-n} }{\alpha_n+1} \frac{1}{r^n}\sum x_i^2 + \frac{d\eps^{-(n+m)}r^m}{\alpha_n+1}\frac{1}{r^n}\sum x_i^2 \right. \\
	&\pushright{ \left. {}  - \frac{d\eps^{-(n+m)}}{\alpha_{n+m}+1} \frac{1}{r^{n+m}} \sum\EE{x_{ij}^2}[x^n] \right\} } \\
	&= \eps^nT_\text{app}^n - \frac{d\eps^{-n}}{\alpha_n+1} \frac{1}{r^n}\sum x_i^2\loo.
\end{align*}
As for the Gaussian field,
\[ 
\mc T_\text{shift}^{n,m} 
	= \frac{\eps^{-m}r^m - 1}{r^m-1} r^{n+m}(1-r^{-m})\sigma_0 
	= r^n\sigma_0(\eps^{-m}r^m-1)
\] 
and we define
\[ 
T_\mathrm{ren}^n = \lim_{m\rightarrow\infty} \EE{\eps^{n+m}T_\mathrm{app}^{n+m} - d\sigma_0\eps^{-n-m}r^{n+m}}[x^n] 
	= \eps^nT_\mathrm{app}^n - d\sigma_0\eps^{-n}r^n.
\] 
Finally, we mention that for the Cauchy field  it is also possible to put the kinetic energy in Wick order, even if it does not satisfy \autoref{renormalizability}. We  leave the calculation to the interested reader, for it does not have any relevance here.

\section{Calculation of the effective Lagrangian}

Now we consider the problem of the existence (and compatibility) of the effective measures
\[
\lim_{m\rightarrow\infty} \EE{\e^{-L_\text{ren}^{n+m}}}[x^n]
\]
for some given first-order renormalized, remaining\footnote{Recall that part of the full Lagrangian has been absorbed into the reference measure.} Lagrangian $L_\text{ren}$\lo. 
We would like this to be given by $\e^{-L_\text{ren}^n}$ modified by some non-local corrections, which turn out to be best understood as multiplicative (as opposed to additive)---but of course are additive at the level of the Lagrangian.
We work first at a formal power series level.
So, we want an $L_\mathrm{eff}^n[\lambda]$ such that
\[
\e^{L_\mathrm{eff}^n[\lambda]} = \lim_{m\rightarrow\infty} \EE{\e^{\lambda L_\mathrm{ren}^{n+m}}}[x^n].
\]
Clearly, $L_\mathrm{eff}[\lambda] = \lambda L_\mathrm{ren} + \text{higher-order terms}$.
\begin{rk}
If the corrections diverge as $m\rightarrow\infty$, higher-order renormalization would be needed.
\end{rk}

Let us implement this. Define the Lagrangian (conditional) cumulants by
\[
L_{\mathrm{eff},k}^{n,m}(x^n) = \left.\frac{\D^k}{\D\lambda^k}\right|_{\lambda=0} \log \EE{\e^{\lambda L_\text{ren}^{n+m}}}[x^n],
\]
so that
\[
\EE{\e^{\lambda L_\text{ren}^{n+m}}}[x^n] = \exp\left( \sum_{k=1}^\infty \frac{\lambda^k}{k!}  L_{\mathrm{eff},k}^{n,m}(x^n) \right) = \e^{\lambda L_\text{ren}^n(x^n)} \exp\left( \sum_{k=2}^\infty \frac{\lambda^k}{k!} L_{\mathrm{eff},k}^{n,m}(x^n) \right)
\]
and the $k$-th non-local correction factor is $\exp\bigl(\frac{\lambda^k}{k!} L_{\mathrm{eff},k}^{n,m}(x^n)\bigr)$.
It is known that cumulant combinatorics are such that
\begin{equation} \label{eq:Ltilde}
\EE{\left( L_\text{ren}^{n+m} \right)^k}[x^n] = \sum_{\text{partitions of } k} L_{\mathrm{eff},k_1}^{n,m}\,\cdots\, L_{\mathrm{eff},k_p}^{n,m}
\end{equation}
and, conversely,
\begin{align*}
L_{\mathrm{eff},k}^{n,m} 
	&= \EE{\left( L_\text{ren}^{n+m} \right)^k}[x^n] - \sum_{\ell=1}^{k-1} {k-1\choose \ell-1} L_{\mathrm{eff},\ell}^{n,m}\, \EE{\left( L_\text{ren}^{n+m} \right)^{k-\ell}}[x^n] \\
	&= \sum_{\text{partitions of } k} (p-1)!(-1)^{p-1} \EE{\left(L_\text{ren}^{n+m}\right)^{k_1}}[x^n] \cdots \EE{\left(L_\text{ren}^{n+m}\right)^{k_p}}[x^n].
\end{align*}

\begin{rk}
By Dyson's argument, the generating function of the Lagrangian cumulants is only formally defined (the power series will not be convergent). Note, however, that we can well have convergence for the full non-local correction it entails.
\end{rk}

From now on, we focus on the case of a reference measure in physical space coordinates perturbed by a quadratic remaining Lagrangian, such as the kinetic energy plus a mass term. 

\subsection{Connected graph expansion for quadratic perturbations}

Consider a Lagrangian of the form
\(
L_\text{ren}^n(x^n) = \frac1{r^n} \sum_{i,i'} c^n_{ii'}x_ix_{i'}\loo.
\)
In order to treat more general Lagrangians, one would have to use hyper-graphs whose hyper-edge vertices are decorated with an integer (the exponent of the corresponding $x_i$).

\begin{defn}
Fix a system $\cdots \rightarrow X^2 \rightarrow X^1 \rightarrow X^0$ whose connecting maps are associated with progressively refined, orthogonal and complete families of projections 
\[
P^n = \set{p_{i_1\dots i_n} | i_k=1\dots r} \subseteq L^\infty(S).
\]
An \emph{edge of resolution $n$} is an unordered pair
\(
e\in E^n := P^n\times P^n/(p_0\lo,p_1)\sim(p_1\lo,p_0).
\)
The equivalence class of $(p_0\lo,p_1)$ in $E^n$ will be written $\Lbag p_0\lo,p_1\Rbag$.
A \emph{(multi-)graph of size $k$ and resolution $n$} is an element of the free, commutative, graded semigroup 
\[
G^n=\bigcup_{k\in\N} G^{n,k},\quad G^{n,k} = \Set{ e_1\cdots e_k | e_i\in E^n }
\] 
generated by $E^n$. 
Now fix $n, m\in\N$ and take $e = \Lbag p_0\lo,p_1\Rbag\in E^{n+m}$. 
The \emph{coarse version (of resolution $n$)} of $e$ is the edge 
\[
P^n\bigl(\Lbag p_0\lo,p_1\Rbag\bigr) = \Lbag q_0\lo,q_1\Rbag\in E^n,\quad q_0 p_0 = p_0\ \text{ and }\ q_1 p_1 = p_1\lo.
\] 
Observe that these conditions determine $q_0$ and $q_1$ uniquely.
In the same way there is a coarse version $P^n(g)\in G^n$ of every graph $g\in G^{n+m}$, namely
\[
P^n(e_1\cdots e_k) = P^n(e_1)\cdots P^n(e_k).
\]
We say that a graph is \emph{(path) connected} if it is connected, in the sense that for any two projections $p, p'$ which are endpoints of edges in the graph, there exists a sequence of adjacent edges connecting them. Correspondingly, we say that $g\in G^{n+m}$ is \emph{coarsely connected,} with $n$ and $m$ understood from context, if $P^n(g)$ is connected. 
Write 
\[
G^{n+m,k}_{\mathrm{c}|n} = \Set{ g\in G^{n+m,k} | P^n(g) \text{ is connected} },\quad G^{n,k}_\mathrm{c} = G^{n,k}_{\mathrm{c}|n}\LOO.
\]
Finally, a \emph{partition} (or \emph{factorization}) of $g$ is a partition of its multiset of edges; thus, a partition is a collection of graphs $g_1\lo,\dots,g_p$ such that $g = g_1 \cdots g_p$\lo.
We will write $h\leq g$ if $h$ is a factor of $g$ and $h<g$ if $h\leq g$ and $h\neq g$.
\end{defn}

Our quadratic Lagrangian determines \emph{adjacency matrices} $c^n: E^n\rightarrow\R$ by $c^n\bigl(\Lbag p_i\lo,p_{i'}\Rbag\bigr) = c^n_{ii'}$\LO. 
Given $g = e_1\cdots e_k\in G^n$, write
\[
c^n(g) = \prod_{i=1}^k c^n(e_i),\quad x^n(g) = \prod_{i=1}^k x^n(e_i)
\]
where $x^n\bigl(\Lbag p_i\lo, p_{i'}\Rbag\bigr) = x_ix_{i'}$\lo. With this conventions,
\[
L_\mr{ren}^n(x^n) = \frac1{r^n}\sum_{e\in E^n} c^n(e)x^n(e).
\]
Observe that 
\[
c^{n}(gg') = c^{n}(g)c^{n}(g'),\quad x^n(gg') = x^n(g)x^n(g'). 
\]
Also, for coarsely disconnected $g$ (i.e.\ $P^n(g)$ disconnected), say $g=g_1g_2$ with $g_1$ a coarsely connected component, one has
\[
\EE{x^{n+m}(g_1g_2)}[x^n] = \EE{x^{n+m}(g_1)}[x^n] \EE{x^{n+m}(g_2)}[x^n].
\]

\begin{defn}
Let $\C[x^{n}]$ be the algebra of polynomials on the variables $\set{ x_{i_1\dots i_n} }$. We define the applications $\chi^n, \chi_\mr{c}^n: G^{n+m}\rightarrow \C[x^n]$ by 
\[
\chi^n(g) = \EE{x^{n+m}(g)}[x^n],\quad 
\chi_\mr{c}^n(g) = \sum_{\text{partitions}} (p-1)!(-1)^{p-1} \chi^n(g_1) \cdots \chi^n(g_p).
\]
\end{defn}

\begin{prop} \label{Ec recurrence}
Let $g\in G^{n+m}$ and choose an edge $e\leq g$.
The following recursive relation holds:
\[ 
\chi^n_\mr{c}(g)
	= \chi^n(g) - \sum_{e\leq h<g} \chi_\mr{c}^n(h) \chi^n(g/h).
\] 
\end{prop}
\begin{proof}
This follows from the fact that
\begin{align*}
&\chi^n(g) - \sum_{\substack{\text{graphs } h<g\\ \text{with } e\leq h}} \sum_{\substack{\text{partitions of } h\\ \text{ with, say, } e\leq h_1}} (p-1)!(-1)^{p-1} \chi^n(h_1) \underbrace{\chi^n(h_2) \cdots \chi^n(h_p) \chi^n(g/h)}_{\mathclap{\substack{g/h \text{ is one of these $p$ factors, so}\\ \text{this product shows up $p$ times} }}} \\
	&\qquad= \chi^n(g) + \sum_{\substack{\text{partitions of } g \text{ with}\\ \text{at least 2 factors}}} p!(-1)^p \chi^n(g_1)\cdots \chi^n(g_{p+1}) = \chi^n_\mr{c}(g). \hfill\qedhere
\end{align*}
\end{proof}
\begin{coro}
One has
\[ 
\chi^n(g)
	= \chi_\mr{c}^n(g) + \sum_{e\leq h<g} \chi_\mr{c}^n(h) \chi^n(g/h) 
	= \sum_{\text{partitions}} \chi^n_\mr{c}(g_1)\cdots \chi_\mr{c}^n(g_p).
\] 
\end{coro}
\begin{proof}
This is easily proved by induction on the length of the graph.
\end{proof}
\begin{coro} \label{Ec(disc)=0}
If $g$ is coarsely disconnected, $\chi_\mr{c}^n(g) = 0$. In particular,
\[
\chi_\mr{c}^n(g) = \sum_{\substack{\text{coarsely connected} \\ \text{partitions}}} \chi_\mr{c}^n(g_1)\cdots \chi_\mr{c}^n(g_p).
\]
\end{coro}
\begin{proof}
Let us do this by induction on the size of $g$. If $g=e_1e_2$ is coarsely disconnected,
\[
\chi_\mr{c}^n(e_1e_2) = \chi^n(e_1e_2) - \chi^n(e_1) \chi^n(e_2) = 0.
\]
Now, suppose that $g=g_1g_2$ with $g_1$ a coarsely connected component. Choosing an edge $e|g_1$ and using the inductive hypothesis,
\begin{align*}
\chi_\mr{c}^n(g) 
	&= \chi^n(g_1g_2) - \chi_\mr{c}^n(g_1) \chi^n(g_2) - \sum_{e\leq h<g_1} \chi_\mr{c}^n(h) \chi^n(g_1g_2/h)  \\
	&= \chi^n(g_1g_2) - \left(\chi_\mr{c}^n(g_1) + \sum_{e\leq h<g_1} \chi_\mr{c}^n(h) \chi^n(g_1/h) \right) \chi^n(g_2) \\
	&= \chi^n(g_1g_2) - \chi^n(g_1) \chi^n(g_2) = 0. \qedhere
\end{align*}
\end{proof}

\begin{theo}
One has that
\[ 
L_{\mathrm{eff},k}^{n,m} (x^n)
	= \frac1{r^{k(n+m)}} \sum_{G^{n+m,k}_{\mathrm{c}|n}} c^{n+m}(g) \chi_\mr{c}^n(g).
\] 
\end{theo}
\begin{proof}
Indeed,
\begin{align*}
\EE{\left(L_\text{ren}^{n+m}(x^{n+m})\right)^k}[x^n] 
	&= \frac1{r^{k(n+m)}} \sum_{g\in G^{n+m,k}} c^{n+m}(g)\EE{x^{n+m}(g)}[x^n] \\
	&= \sum_{k_1+\cdots+k_p=k} \prod_{\ell=1}^p \sum_{g_\ell\in G_{\mathrm{c}|n}^{n+m,k_\ell}}  \frac1{r^{k_\ell(n+m)}} c^{n+m}(g_\ell) \chi_\mr{c}^n(g_\ell), 
\end{align*}
from which the result follows by comparison with
~\eqref{eq:Ltilde}.
\end{proof}

\subsection{Graph expansion for a mass perturbation}

From now on, we work with the particular case of a $\Gamma$ reference field. Before getting into the business of adding a kinetic energy term, we consider the simpler situation in which the remaining Lagrangian is just a mass perturbation of the potential energy, i.e.\ has the form
\[
L_\mr{ren}^n(x^n) = \frac1{r^n}\sum \mc V^n_2(x_i) = \frac1{r^n}\sum \frac{\alpha_0 (\alpha_0+1)}{r^{2n} \alpha_n (\alpha_n+1)} x_i^2 = \frac1{r^n}\sum \frac{\alpha_0+1}{\alpha_0+r^n} x_i^2\Lo.
\]
\begin{rk}
Recall~\cite{unser2014introduction} that besides the Gaussian, there is no infinitely divisible distribution with tail
\[
f(x) = \e^{-O(\abs x^{1+\eps})},\quad \eps>0.
\]
Thus, if a mass perturbation is possible, the resulting effective Lagrangian cannot be a polynomial, and coefficients with negative sign must show up infinitely often in its power series expansion. In order to assess the factibility of introducing a mass perturbation, note that, by Jensen's inequality,
\[
\EE{\e^{-L_\mathrm{ren}}} \geq \int \e^{-\EE{L_\mathrm{ren}}[x^0]}f(x^0)\D x^0 >0
\]
and therefore what could fail is $\EE{\e^{-L_\mathrm{ren}}} = \infty$. This is also excluded in the case of a $\Gamma$ reference, for then the Wick polynomials are non-negative functions and therefore $\EE{\e^{-L_\mathrm{ren}}}<1$. 
So, we expect the limit measure to exist. Moreover, as we shall shortly see, the power series for the effective Lagrangian is convergent in this case.
\end{rk}

Define first, for a graph $h\in G^{n+m}$, the real numbers $\tilde \chi^n(h), \tilde \chi^n_\mr{c}(h)$ by the equations
\[
\chi^n(h) = \tilde \chi^n(h) x^n(P^nh),\quad \chi_\mr{c}^n(h) = \tilde \chi_\mr{c}^n(h) x^n(P^nh).
\] 
Note that
this is only possible thanks to the simple structure of the conditional expectation of a monomial on the  $\Gamma$ field---in other cases, lower degree terms will show up. 
With this notation, we can write the effective Lagrangian as
\[
L_\mr{eff}^n(x^n) 
= \sum_k \frac{(-1)^k}{r^nk!} \sum_{g\in G^{n,k}_\mr{c}} c^n_\mr{eff}(g) x^n(g),\quad
c^n_\mr{eff}(g) 
	= \lim_{m\rightarrow\infty} \frac1{r^{mk}}  \sum_{P^nh = g} c^{n+m}(h) \tilde \chi_\mr{c}^n(h), 
\]
where $c^n_{ij}$ equals $\frac{\alpha_0+1}{\alpha_0+r^n}$ if $i=j$ and 0 otherwise. Two problems must be studied: existence of these limits, and convergence of the resulting power series. 

Given the vanishing of non-diagonal $c^n_{ij}$'s, non-zero contributions to the effective Lagrangian come from graphs containing only loops, and those can only be coarsely connected if all the loops belong to the same coarse region. Thus, the effective Lagrangian is given by a density:
\begin{align*}
L^n_\mr{eff}(x^n) &= \sum_k \frac{(-1)^k}{r^nk!} \sum_{\text{loops } e\in E^n} c_\mr{eff}^n(e^k) x^n(e)^k = \sum_k \frac{(-1)^k}{r^nk!} \sum_{i} c_\mr{eff}^n\bigl((p_i\lo,p_i)^k\bigr) x_i^{2k} \\
	&= \frac1{r^n} \sum_i  \left(\lim_{m\rightarrow\infty} \tilde{\mc L}^{n,m}_\mr{eff} \right) x_i^{2k}\Lo, 
\end{align*}
with 
\[
\tilde{\mc L}^{n,m}_\mr{eff} = \sum_k \frac{(-1)^k}{r^{n(k-1)}k!} \underbrace{ \frac1{r^{km}} \sum_{P^nh = (p_i,p_i)^k} c^{n+m}(h) \tilde \chi_\mr{c}^n(h). }_{\tilde{\mc L}^{n,m}_{\mr{eff},k}}
\]
Let us compute: setting $e=(p_i\lo,p_i)$ and choosing a total order $\prec$ on $E^{n+m}$,
\begin{align*}
\tilde{\mc L}_{\mr{eff},k}^{n,m} 
	&= \frac1{r^{km}} \sum_{\substack{f_1\dots f_k \text{ loops} \\ \text{with } P^n(f_i)=e}} c^{n+m}(f_1\cdots f_k) \tilde \chi_\mr{c}^n(f_1\cdots f_k) \\
	&= \frac1{r^{km}} \biggl\{ \sum_{f_1} c^{n+m}(f_1^k) \tilde \chi_\mr{c}^n(f_1^k)  + \sum_{f_1\prec f_2}\sum_{\substack{k_1+k_2=k \\ k_1,k_2\neq 0}}  {k\choose k_1\ k_2} c^{n+m}(f_1^{k_1}f_2^{k_2}) \tilde \chi_\mr{c}^n(f_1^{k_1} f_2^{k_2})  \\
	&\pushright{ \biggl.{} +\cdots +  \sum_{f_1\prec\cdots\prec f_k} k! c^{n+m}(f_1\cdots f_k) \tilde \chi_\mr{c}^n(f_1\cdots f_k) \biggr\} } \\
	&= \frac1{r^{km}} \biggl\{ r^m c^{n+m}(f_1^k) \tilde\chi_\mr{c}^n(f_1^k) + {r^m\choose 2} \sum_{\substack{k_1+k_2=k \\ k_1,k_2\neq 0}} {k\choose k_1\ k_2} c^{n+m}(f_1^{k_1} f_2^{k_2}) \tilde\chi_\mr{c}^n(f_1^{k_1} f_2^{k_2})  \\
	&\pushright{ \biggl.{}+\cdots + {r^m\choose k} k! c^{n+m}(f_1\cdots f_k) \tilde \chi_\mr{c}^n(f_1\cdots f_k) \biggr\} } \\
	&= \sum_{\ell=1}^k \underbrace{ \sum_{\substack{k_1+\cdots+k_\ell=k \\ k_i\neq 0}} \frac1{r^{m\ell}} {r^m\choose\ell} {k\choose k_1\ \cdots\ k_\ell} }_{ \mr{I}_{m,\ell} } \cdot \underbrace{ \frac{c^{n+m}(f_1^{k_1}\cdots f_\ell^{k_\ell})}{r^{m(k-\ell)}} \tilde \chi_\mr{c}^n(f_1^{k_1}\cdots f_\ell^{k_\ell}) }_{ \mr{II}_{m,\ell} }.
\end{align*}
Now, we have
\[
\sum_{\ell=1}^k \sum_{\substack{k_1+\cdots+k_\ell=k \\ k_i\neq 0}} \frac1{\ell!} {k\choose k_1\ \cdots\ k_\ell} = \sum_{\ell=1}^k {k\brace \ell} = p(k),
\]
where ${k\brace\ell}$ is the Stirling partition number, counting the ways of partitioning a set of $k$ elements into $\ell$ subsets, and $p(k)\in\N$ is the number of partitions of $k$. Thus, 
\[
\sum_{\ell=1}^k \lim_{m\rightarrow\infty} \mr{I}_{m,\ell} = p(k) \sim \frac1{4k\sqrt 3} \exp\left(\pi \sqrt{\frac{2k}3}\right).
\]
\begin{prop}
\[
\lim_{m\rightarrow\infty} \mr{II}_{m,\ell} \lesssim \frac{k!}{ 2\left( r^{2n}(\alpha_n)_2\log 2 \right)^k.}
\]
\end{prop}
\begin{proof}
Write
\begin{align*}
&\frac{c(f_1^{k_1}\cdots f_\ell^{k_\ell})}{r^{m(k-\ell)}} \EE[\mathrm c]{x(f_1^{k_1}\cdots f_\ell^{k_\ell})}[x^n] \\
	&\quad= \frac1{r^{2nk}} \frac{(\alpha_0)_2^k}{r^{mk} (\alpha_{n+m})_2^k} \sum_\text{partitions} (p-1)! \frac{ \EE{x(g_1)}[x^n]\cdots \EE{x(g_p)}[x^n] }{r^{m(2k-\ell)}}.
\end{align*}
Next, observe that
\[
\lim_{m\rightarrow\infty} \frac{1}{ r^{mk}(\alpha_{n+m})_2^k } = \lim_{m\rightarrow\infty} \frac{1}{\alpha_n^k(1+\alpha_{n+m})^k} = \alpha_n^{-k}\LO.
\]
As for the sum over partitions, there are two possibilities:
\begin{description}
	\item[Case 1.] The partition does not subdivide any of the connected components of $g$, namely the subgraphs $f_i^{k_i}$\LO. Then,
\begin{align*}
&\frac{r^{m\ell}}{r^{2mk}} \EE{x(g_1)}[x^n]\cdots \EE{x(g_p)}[x^n] \\
	&\quad= \frac{ r^{m\ell}(\alpha_{n+m})_2^{k_1} \cdots (\alpha_{n+m})_2^{k_\ell} }{ (\alpha_n)_{2\abs{g_1}} \cdots (\alpha_n)_{2\abs{g_p}} } 
	\leq \frac{\alpha_n^\ell (1+\alpha_{n+m})^k}{ (\alpha_n)_2^k } \rightarrow \alpha_n^{\ell-k} (1+\alpha_n)^{-k}.  
\end{align*}
	\item[Case 2.] The partition does subdivide some connected components, say factoring $f_i^{k_i} = f_i^{k_i'} f_i^{k_i''}$ for $i=1\dots \ell'\leq \ell$. Then, similarily,
\[ 
\frac{r^{m\ell}}{r^{2mk}} \EE{x(g_1)}[x^n]\cdots \EE{x(g_p)}[x^n] 
	\leq \frac{\alpha_n^\ell \alpha_{n+m}^{\ell'} (1+\alpha_{n+m})^k }{ (\alpha_n)_2^k } \rightarrow 0.
\] 
\end{description}
The claim follows, because 
\[
\sum_{p=1}^k (p-1)! {k\brace p} \leq \sum_{p=1}^k p! {k\brace p} \sim k!/2(\log 2)^{k+1}. \qedhere
\]
\end{proof}

Thus, the power series $\mc L_\mathrm{eff}^n[\lambda]$ has a strictly positive radius of convergence which, moreover, increases with $n$. 

\subsection{Divergence of the $\Gamma$ field effective kinetic energy}

Now we take
\[ 
L_\mr{ren}^n(x^n) 
	= T_\mr{ren}^n(x^n) = \frac1{r^n} \sum c^n_{ii'} x_ix_{i'}\lo,\quad c_{ii'}^n = \left\{\begin{aligned} &\frac{d\alpha_n\eps^{-n}}{1+\alpha_n}& &i=i', \\ &-\eps^{-n}& &i,i'\text{ nearest neighbours,} \\ &0& &\text{otherwise.} \end{aligned}\right.
\] 
We will see that there are divergences in the resulting effective Lagrangian. 

We take a first step towards identifying potentially divergent contributions in $c_\mr{eff}^n(g)$ by noting that we can focus on 
\[
b_\mr{eff}^n(g) 
	= \lim_{m\rightarrow\infty} \frac1{r^{mk}} \sum_{P^nh=g} c^{n+m}(h) \tilde \chi^n(h).
\]
Indeed, we have the following.
\begin{prop}
If the $b_\mr{eff}^n(g)$'s exist, then the $c_\mr{eff}^n(g)$'s exist too, and one has
\[
c^n_\mr{eff}(g)
	= b_\mr{eff}^n(g) - \sum_{e\leq h< g} c_\mr{eff}^n(h) b_\mr{eff}^n(g/h). 
\]
\end{prop}
\begin{rk}
Of course, to establish divergence rigorously we need the reciprocal, but we choose to dispense with that for the time being.
\end{rk}
\begin{proof}
Let us do this by induction on the number of edges.
Let $g$ be a graph of size $k$ and $e$ an edge such that $ge$ is connected.
Want to prove that
\[
\lim_{m\rightarrow\infty} \frac1{r^{m(k+1)}} \sum_{P^nh=g} \sum_{P^nf=e} c^{n+m}(hf) \tilde\chi_\mr{c}(hf)
\]
exists. Using \autoref{Ec recurrence}, have that
\begin{align*}
&\frac1{r^{m(k+1)}} \sum_{P^nh=g} \sum_{P^nf=e} c^{n+m}(hf) \tilde \chi_\mr{c}(hf) \\
	&\quad= \frac1{r^{m(k+1)}} \sum_{P^nh=g} \sum_{P^nf=e} c^{n+m}(hf) \left( \tilde\chi(hf) - \sum_{f\leq h'<hf} \tilde \chi_\mr{c}(h') \tilde\chi(hf/h') \right), 
\end{align*}
and
\begin{align*}
&\lim_{m\rightarrow\infty} \frac1{r^{m(k+1)}} \sum_{P^nh=g} \sum_{P^nf=e}  \sum_{f\leq h'<hf} c^{n+m}(hf) \tilde \chi_\mr{c}(h') \tilde \chi(hf/h') \\
	&\quad= \lim_{m\rightarrow\infty} \frac1{r^{m(k+1)}} \sum_{e\leq g'<ge} \sum_{P^nh'=g'} \sum_{P^nh''=ge/g'} c^{n+m}(h'h'') \tilde \chi_\mr{c}(h') \tilde \chi(h'') \\
	&\quad= \sum_{e\leq g'<ge} \lim_{m\rightarrow\infty} \sum_{P^nh'=g'} \frac1{r^{m\abs{h'}}} c^{n+m}(h') \tilde \chi_\mr{c}(h') \sum_{P^nh''=ge/g'} \frac1{r^{m\abs{h''}}} \tilde \chi(h''). \qedhere
\end{align*}
\end{proof}

Let us focus on $b^n_\mr{eff}(e^k)$ with $e$ a loop over a fixed coarse site $i$ and $k=1,2$. 
Now, suppose that a choice of identification
\[
\Set{ ij | j=1,\dots, r^m } \cong \Lambda := \set{ 1,\dots, \eps^{-m} }^{d}
\]
respecting the geometry (i.e.\ such that neighouring sites are mapped to neighbouring sites) has been made, and define
\[
\begin{aligned} \sigma_\ell : D_\ell\subseteq \Lambda &\rightarrow \Lambda \\ \lambda &\mapsto (\lambda_1\lo,\dots,\lambda_\ell+1,\dots \lambda_d) \end{aligned}
\]
where $D_\ell = \Set{\lambda\in \Lambda | \lambda_\ell<\eps^{-m}}$. Let also 
\begin{align*}
\Lambda_d &= D_{1}\cap\cdots\cap D_{d} 
	= \Set{ \lambda\in \Lambda |(\forall\ell)\ a_\ell<\eps^{-m} } \\
\Lambda_{d-1} &= \Set{ \lambda\in \Lambda | (\exists!\ell)\ \lambda_\ell=\eps^{-m} } 
	= \Set{ \lambda\in \Lambda | \#\set{\ell| \lambda_\ell<\eps^{-m}} = d-1 } \\
&\,\ \vdots \\
\Lambda_1 &= 
	\Set{ \lambda\in \Lambda | \#\set{\ell|\lambda_\ell<\eps^{-m}}=1 } \\
\Lambda_0 &= \set{(\eps^{-m},\dots,\eps^{-m})},
\end{align*}
so that $\Lambda=\bigcup_{\ell=0}^d \Lambda_\ell$\lo. The idea here is to add together the contribution of a loop over $ij$ with that of the edges of the form $(ij,i\sigma_\ell(j))$, for this localizes the cancellations that occur. Let us exemplify this by revisiting the example of $c_\mr{eff}^n(e) = b_\mr{eff}^n(e)$, which was done when renormalizing the kinetic energy. We have---abusing notation, for one still has to take $\lim_{m\rightarrow\infty}$ in the RHS,
\begin{align*}
b^n_\mr{eff}(e)
	&= \frac1{r^m} \left\{ \sum_{j\in \Lambda_d} \left(c_{jj}^{n+m} \tilde \chi\bigl((j,j)\bigr) + \sum_{\ell=1}^d c_{j\sigma_\ell(j)}^{n+m} \tilde \chi\bigl( (j,\sigma_\ell(j)) \bigr) \right) + \cdots \right\} \\
	&= \frac1{r^m} \left\{ \sum_{j\in \Lambda_d} \left( \frac{d\alpha_{n+m}\eps^{-n-m}}{1+\alpha_{n+m}} r^{2m}\frac{(\alpha_{n+m})_2}{(\alpha_n)_2} - d\eps^{-n-m}r^{2m}\frac{\alpha_{n+m}^2}{(\alpha_n)_2} \right) + \cdots \right\} \\
	&= \frac1{r^m} \frac{\alpha_n}{1+\alpha_n} \left\{ \sum_{j\in \Lambda_d} \eps^{-n-m}(d-d) + \sum_{j\in \Lambda_{d-1}} \eps^{-n-m}\bigl(d-(d-1)\bigr) + \cdots\right\} \\
	&= \frac1{r^m} \frac{\alpha_n}{1+\alpha_n} \left\{ \sum_{j\in \Lambda_{d-1}} \eps^{-n-m} + \sum_{j\in \Lambda_{d-2}} 2\eps^{-n-m} + \cdots \right\}
\end{align*}
Thus, knowing that $\Abs{\Lambda_\ell} = {d\choose\ell} (\eps^{-m}-1)^\ell = {d\choose\ell} r^m\eps^{m(d-\ell)}(1 + O(\eps^m))$, only the first term survives, giving the expected result. 

The calculation above was done in \autoref{kinetic energy renormalization} in a way that can be diagrammatically represented as:
\[ 
b^n_\mr{eff}(e)
= \loopdiagram{} + \branchdiagram*{}[] 
	= \frac1{r^m} \sum_{j\in \Lambda} \loopdiagram{j} + \frac1{r^m}\sum_{\ell=1}^d \sum_{j\in \Lambda_\ell} \branchdiagram*{j}
\] 
whereas now we are performing the sum as follows:
\begin{align*}
b^n_\mr{eff}(e)
	&= \frac1{r^m} \sum_{j\in \Lambda_d}\left( \loopdiagram{j} + \branchdiagram{j}[1,,d] \right) + \frac1{r^m}\sum_{j\in \Lambda_{d-1}}\left( \loopdiagram{j} + \branchdiagram{j}[1,d-1] \right) + \cdots \\
	&= \left(\loopdiagram{} + \branchdiagram{}[1,,d]\right) + \left(\loopdiagram{} + \branchdiagram{}[1,d-1]\right) + \cdots + \left(\loopdiagram{} + \branchdiagram*{}\right) + \loopdiagram{r^m} \\
	&= \left( \loopdiagram{} +\branchdiagram{}[1,d-1] \right) + O(\eps^m).
\end{align*}
Here, diagrams with vertices which are not fully specified represent sums over all their possible concrete realizations, and the vertex $r^m$ is the only element of $\Lambda_0$\lo. Some words of caution regarding the interpretation of this picture:
\begin{itemize}
	\item A branching diagram $\branchdiagram{}[,]$ must not be interpreted as one term with several propagators, but as a sum of several terms with a single propagator. Thus, rather than representing a situation with more than two vertices, it is representing the sum over all possible situations with two vertices, with the right ones being (all and only) those reachable by application of a $\sigma_\ell$ to the left one, which we will refer to as the \emph{source.} 
	\item The parenthesis in the expression $\left( \loopdiagram{}+\branchdiagram{}[] \right)$ does not just indicate operation precedence: it is also implied that the vertex in the loop is the same as the source vertex in the branching. 
\end{itemize}
It will be convenient to write this even more succintely, as follows:
\begin{align*}
\diagram{\bouquet{}[1,\ell]}
	&= \left(\loopdiagram{} + \branchdiagram{}[1,\ell]\right) 
	= \frac1{r^m}\sum_{j\in \Lambda_\ell} \left( \loopdiagram{j} + \branchdiagram{j}[1,\ell] \right) \\
	&= \frac1{r^m} {d\choose\ell}(\eps^{-m}-1)^\ell \frac{\alpha_n\eps^{-n-m}}{1+\alpha_n}(d-\ell),
\end{align*}
which converges to 0 unless $\ell=d-1$. We will call this a \emph{bouquet} diagram, whose \emph{source} is the site over which the loop stands.

When doing the same with bouquets of several loops, internal propagators can produce finely connected non-trivial subdiagrams that must be dealt with separately. Thus, for instance, for computing $b^n_\mr{eff}(e^2)$
we group the corresponding sum over fine diagrams according to the following partition of the set of possible choices:
\begin{align*}
&\Set{ (f,f') | P^n(ff') = \eightdiagram } \\
	&\quad= \Set{ (f,f') | s(f)=s(f') } \\
	&\qquad{} \cup \Set{ (f,f') | (\exists \ell)\  s(f')=\sigma_\ell\bigl( s(f) \bigr) } \cup \Set{ (f,f') | (\exists \ell)\ s(f)=\sigma_\ell\bigl( s(f') \bigr) } \\
	&\qquad{}\cup \Set{ (f,f') | (\exists\ell,\ell')\ \sigma_\ell\bigl(s(f)\bigr) = \sigma_{\ell'}\bigl( s(f')\bigr) } \\
	&\qquad{}\cup \Set{\text{the rest}}.
\end{align*}
The resulting sum can be diagrammatically expressed as follows:
\[ 
b^n_\mr{eff}(e^2)
	= \bouquetdiagram{2}[,]
	+ 2\,\diagram{\bouquet{1}[,] \bouquet(\propagatorlength,0){1}[,]} 
	+ \diagram{\bouquet{1}[,] 	\bouquet(\propagatorlength,-\propagatorlength){1}[,]}
	+ \bouquetdiagram{1}[,]\ \bouquetdiagram{1}[,]
\] 
where, now, the label of the source vertex in a bouquet diagram indicates the number of branches that must be choosen (with repetitions allowed). 
Now, for diagrams having one, as opposed to two, connected components, the cancellation between positive and negative terms that prevented $b^n_\mr{eff}(e)$ from diverging will be spoiled. Thus, one should not expect to have convergence unless each connected diagram is convergent on its own. This, however, is not the case: 
take, for instance,
\begin{align*}
\bouquetdiagram{2}[,]
	&= \frac1{r^{2m}} \sum_{\ell=0}^d\ \bouquetdiagram{2}[1,\ell] \\
	&= \frac1{r^{2m}} \sum_{\ell=0}^d\sum_{j\in \Lambda_\ell}\left( \eightdiagram[j] + \ell\,\tadpolediagram[j] + \ell(\ell-1)\elldiagram[j] + \ell\ \eyediagram[j]\right).
\end{align*}
We assess the asymptotics of this expression as follows:
\begin{itemize}
	\item As noted above, a sum over $\Lambda_\ell$ has $O(r^m\eps^{m(d-\ell)})$ terms.
	\item Each edge constributes a factor of $\eps^{-m}$, with loops contributing also an extra $r^{-m}$.
	\item Inspecting the explicit expression for $\EE{(x_{ij_1})^{k_1}\cdots (x_{ij_p})^{k_p}}[x^n]$ obtained in \autoref{cond_exp_lem}, we see that each vertex $j$ contributes a factor $r^{m(\deg j-1)}$, where $\deg j$ is the number of incident edges (counted with multiplicity, i.e.\ loops incide twice).
\end{itemize}
Thus, for instance, the ``eye'' diagram is divergent, with
\[
\frac1{r^{2m}} \sum_{j\in\Lambda_\ell} \eyediagram[j]
\]
being of order $r^m\eps^{m(d-\ell)} \eps^{-2m} = \eps^{-m(2+\ell)}$.
Worse divergences appear as one increases the number of loops in the bouquet. Thus, it seems unlikely that the theory can be renormalized, i.e.\ that a finite number of higher-order counterterms can render the effective kinetic energy finite.

\providecommand{\bysame}{\leavevmode\hbox to3em{\hrulefill}\thinspace}
\providecommand{\MR}{\relax\ifhmode\unskip\space\fi MR }
\providecommand{\MRhref}[2]{%
  \href{http://www.ams.org/mathscinet-getitem?mr=#1}{#2}
}
\providecommand{\href}[2]{#2}

\end{document}